\documentclass[twocolumn, superscriptaddress,pra,aps,showpacs]{revtex4-1}
\usepackage{amsmath,mathrsfs,amsbsy,color,graphicx,bm,amsthm,amsfonts,array,multirow}
\usepackage[english]{babel}
\usepackage[latin1]{inputenc}
\usepackage{graphicx}
\usepackage{braket}
\usepackage{tabularx}
\usepackage{afterpage}

\newtheorem{principle}{Principle}

\newtheorem{proposition}[principle]{Proposition}
\newtheorem{corollary}[principle]{Corollary}
\newtheorem{definition}[principle]{Definition}
\newtheorem{theorem}[principle]{Theorem}
\newtheorem{lemma}[principle]{Lemma}

\newcommand{\Tr}{\text{Tr}}
\begin{document}

\title{Qudit hypergraph states}

\author{F. E. S. Steinhoff} 
\email{steinhofffrank@gmail.com}
\affiliation{Naturwissenschaftlich Technische Fakult\"at, 
Universit\"at Siegen, Walter-Flex-Str. 3, D-57068 Siegen, Germany}
\affiliation{Instituto de F\'isica, Universidade Federal de Goi\'as, 74001-970, Goi\^ania, Goi\'as, Brazil}
\affiliation{Instituto de Engenharia, Universidade Federal de Mato Grosso, 78060-900 V\'arzea Grande, Mato Grosso, Brazil} 
\author{C. Ritz}
\author{N. I. Miklin}
\author{O. G\"uhne}
\affiliation{Naturwissenschaftlich Technische Fakult\"at, Universit\"at Siegen, Walter-Flex-Str. 3, D-57068 Siegen, Germany}

\pacs{03.65.Ud, 03.67.Mn}

\date{\today}

\begin{abstract} 
We generalize the class of hypergraph states to multipartite systems of 
qudits, by means of constructions based on the $d$-dimensional Pauli 
group and its normalizer. For simple hypergraphs, the different equivalence 
classes under local operations are shown to be governed by a greatest 
common divisor hierarchy. Moreover, the special cases of three qutrits 
and three ququarts is analysed in detail.  
\end{abstract}

\maketitle

\section{Introduction}

The physical properties of multipartite systems are highly relevant 
for practical applications as well as foundational aspects. Despite 
their importance, multipartite systems are in general very complex 
to describe and little analytical knowledge is available in the 
literature. Well-known examples in many-body physics are the various
spin models, which are simple to write down, but where typically not 
all properties can be determined analytically. The entanglement properties 
of multipartite systems are no exception and already for pure states 
it is known that a complete characterization is, in general, not a 
feasible task \cite{multientang, gtreview}. This motivates the adoption 
of simplifications that enable analytical results or at least to infer 
properties in a numerically efficient way. 

One approach in this direction with broad impact in the literature is that 
based on a graph state encoding \cite{graphstates}. Mathematically, a graph
consists of a set of vertices and a set of edges connecting the vertices.
Graph states are a class 
of genuinely multipartite entangled states that are represented by graphs. 
This class contains as a special case the whole class of cluster states, 
which are the key ingredients in paradigms of quantum computing, e.g., the 
one-way quantum computer \cite{1wayqc} and quantum error correction 
\cite{qecc} or for the derivation of Bell inequalities 
\cite{bellgraph}. Interestingly, results and techniques of the mathematical theory of graphs can be translated into the graph state framework: one prominent example is the graph operation known as local complementation. The appeal of graph states comes 
in great part from the so-called stabilizer formalism \cite{qecc}. The stabilizer 
group of a given graph state can be constructed in a simple way from local Pauli operators and is abelian; the stabilizer operators associated to a given 
graph state are then used in a wide range of applications such as quantum error 
correcting codes \cite{qecc}, in the construction of Bell-like theorems \cite{bellgraph}, entanglement witnesses \cite{ewgraphs}, models of topological quantum computing \cite{toriccode} and others. 

Recently, there has been an interest in the generalization of graph states to 
a broader class of states known as \textit{hypergraph states} \cite{hypergraphpapers}. 
In a hypergraph an edge can connect more than two vertices, so hypergraph 
states are associated with many-body interactions beyond the usual two-body 
ones. Interestingly, the mathematical description of hypergraph states is 
still very simple and elegant and  in Ref.~\cite{hyper1} a full classification 
of the local unitary equivalence classes of hypergraphs states up to four qubits 
was obtained. Also, in Refs.~\cite{hyper1, hyper2} Bell and Kochen-Specker 
inequalities have been derived and it has been shown that some hypergraph 
states violate local realism in a way that is exponentially increasing with 
the number of qubits. Finally, recent studies in condensed matter theory showed 
that this class of states occur naturally in physical systems associated with 
topological phases \cite{yoshida}.  Originally, hypergraph states were defined 
as members of an even broader class of states known as locally maximally 
entangleable (LME) states \cite{lme}, which are associated to 
applications such as quantum fingerprinting protocols \cite{fingerprint}. 
Hypergraph states are then known as $\pi$-LME states and display the main 
important features of the general class of LME states.  

Up to now, hypergraph states were defined only in the multi-qubit setting, 
while graph states can be defined in systems with arbitrary dimensions. 
In higher dimensions graph states have many interesting properties not 
present in the two-dimensional setting. For example, there are considerable 
differences between systems where the underlying local dimensions are a prime 
or non-prime \cite{vourdas}. Another difference is the construction of 
Bell-like arguments for higher-dimensional systems \cite{quditbellgraph}.  

In the present work, we extend the definition of hypergraph states 
to multipartite systems of arbitrary dimensions (qudits) and analyse 
their entanglement properties. Especially, we focus on the equivalence 
relations under local unitary (LU) operations or under stochastic local 
operations assisted by classical communication (SLOCC). In particular, 
the possible local inter-conversions between different entangled 
hypergraph states are governed by a greatest-common-divisor hierarchy.  
Note that the whole class of qudit graph states is a special case of 
our formulation. 

The paper is organized as follows: In Section II we start by giving 
a brief review of the concepts and results that are at the basis of 
our formulation. This includes a description of the Pauli and Clifford 
groups in a $d$-dimensional system, as well as a general look on qudit 
graph states. In Section III we introduce the definitions associated with qudit hypergraph 
states. Section IV presents some properties of the stabilizer formalism used for qudit
hypergraph states. Section V introduces  the problem of classifying the SLOCC and LU classes of hypergraph states, first describing the different techniques employed 
and then proving a series of results on this classification.  Finally, we present 
some concrete examples in low dimensional tripartite systems, where already the 
main differences between systems of prime and non-prime dimensions become apparent. 
We reserve to the Appendices the related subjects of a phase-space description and 
local complementation of qudit graphs. 

\section{Background and basic definitions}

We consider an $N$-partite system $\mathcal{H}=\bigotimes_{i=1}^N\mathcal{H}_i$, where the subsystems $\mathcal{H}_i$ have the same dimension $d$. A graph is a pair $G=(V,E)$, where $V$ is the set of vertices and $E$ is a set comprised of $2$-element subsets 
of $V$ called edges. Likewise, a hypergraph is a pair $H=(V,E)$, where $V$ are the vertices and $E$ is a set comprised of subsets of $V$ with arbitrary number 
of elements; a $n$-element $e\in E$ is called a $n$-hyperedge. In some sense, 
a hyperedge is an edge that can connect more than two vertices. A multi-(hyper)graph
is a set where the (hyper)edges are allowed to appear repeated. An example of 
a multi-graph can be found in Fig.~\ref{fig:graph}, while one of a multi-hypergraph can be found in Fig.~\ref{fig:hgraph}. Given two integers $m$ and $n$, their greatest common divisor will be denoted by $gcd(m,n)$. The integers modulo $n$ will be denoted as $\mathbb{Z}_n$.  

\begin{figure}[t]
\begin{center}
\includegraphics[scale=.3]{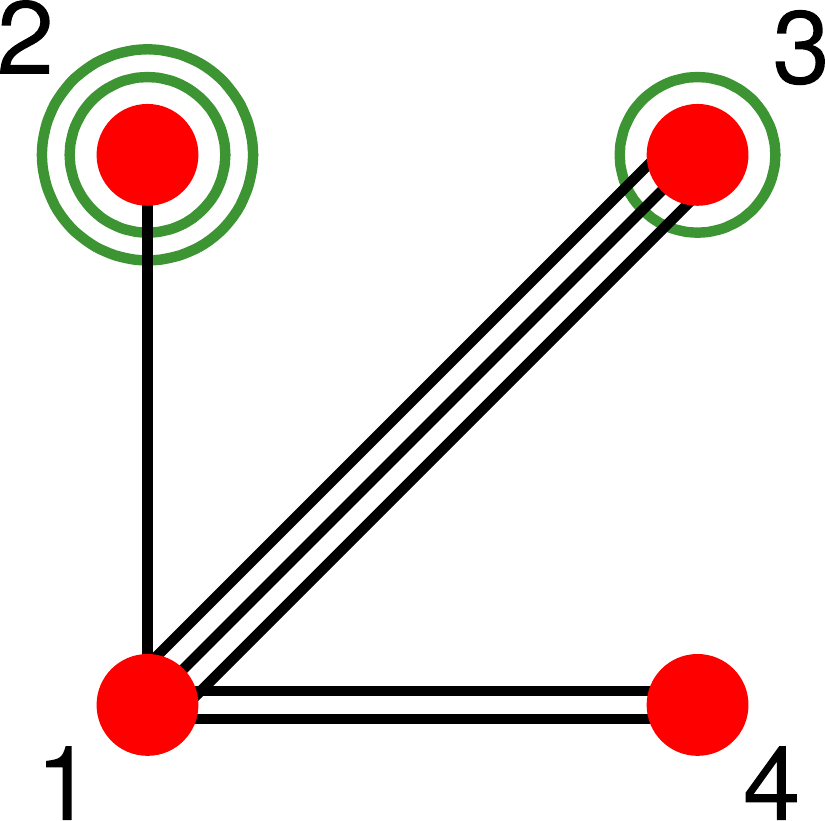} 
\end{center}
\vspace{-.2cm}
\caption{Example of a graph state represented by the multi-graph $G=(V,E)$, where $V=\{1,2,3,4\}$ and $E=\{\{1,2\},\{1,3\},\{1,3\},\{1,3\},\{1,4\},\{1,4\},\{2\},\{2\},\{3\}\}$. The graph state in this case is $|G\rangle=Z_{12}Z_{13}^3Z_{14}^2Z_2^2Z_3|+\rangle^V$.}
\label{fig:graph}
\end{figure}

\begin{figure}[t]
\begin{center}
\includegraphics[scale=.3]{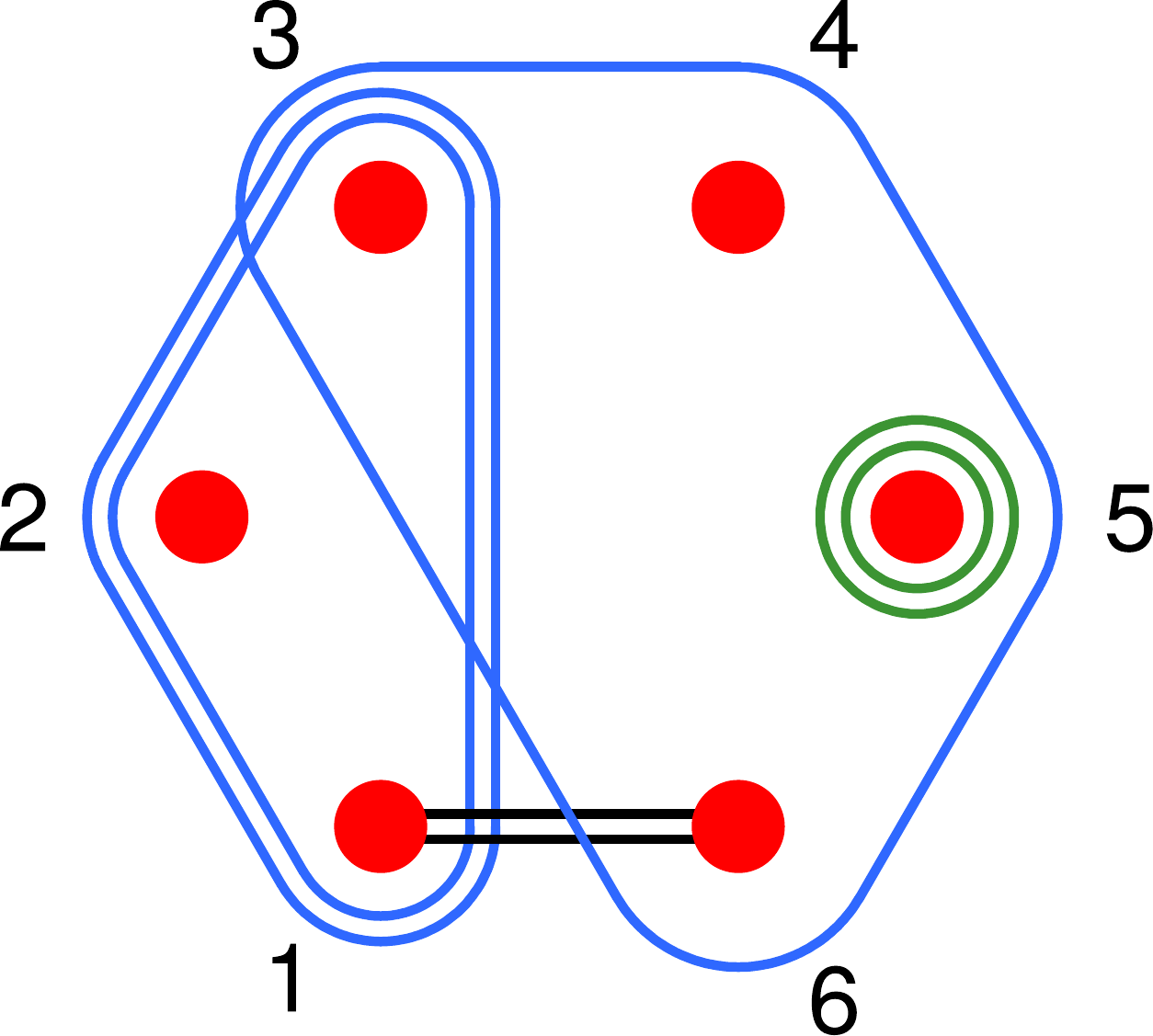} 
\end{center}
\vspace{-.2cm}
\caption{Hypergraph state represented by the multi-hypergraph $H=(V,E)$, with $V=\{1,2,3,4,5,6\}$ and $E=\{\{1,2,3\},\{1,2,3\},\{1,6\},\{1,6\},\{5\},\{5\},\{3,4,5,6\}\}$. The corresponding hypergraph state is then $|H\rangle=Z_{123}^2  Z_{16}^2 Z_5^2 Z_{3456}|+\rangle^V$.}
\label{fig:hgraph}
\end{figure}

\subsection{The Pauli group and its normalizer}

Taking inspiration in the formulation of qubit hypergraph states, we adopt here the description based on the Pauli and Clifford groups in finite dimensions. In a $d$-dimensional system with computational basis $\{|q\rangle\}_{q=0}^{d-1}$,  let us consider the unitary operators given by
\begin{eqnarray}
Z = \sum_{q=0}^{d-1}\omega^q|q\rangle\langle q|; \ X = \sum_{q=0}^{d-1}|q\oplus 1\rangle\langle q|
\end{eqnarray}
with the properties $X^d=Z^d=I$ and $X^mZ^n=\omega^{-mn}Z^nX^m$, where $\omega=e^{2\pi i/d}$ is the $d$-th root of unity and $\oplus$ denotes addition modulo $d$. 
The group generated by these operators is known as the Pauli group and the operators $X^{\alpha}Z^{\beta}$, for $\alpha,\beta\in\mathbb{Z}_d$ are refereed as Pauli operators.   
For $d=2$ these operators reduce to the well-known Pauli matrices for qubits.
In general, these operators enable a phase-space picture for finite-dimensional systems, 
via the relations $Z=e^{\frac{2\pi i}{d}  Q}$, $X=e^{-\frac{2\pi i}{d} P}$, where $Q=\sum_{q=0}^{d-1} q|q\rangle\langle q|$ and $P=\sum_{q=0}^{d-1}q|p_q\rangle\langle p_q|$ are discrete versions of the position and momentum operators; 
here $|p_q\rangle=F|q\rangle$ and $F=d^{-1/2}\sum_{q',q=0}^{d-1} \omega^{q'q}|q'\rangle\langle q|$ is the discrete Fourier transform. 
Thus, $X$ performs displacements in the computational (position) basis, while $Z$ performs displacements in its Fourier transformed (momentum) basis. 

Another set of important operators are the so-called Clifford or 
symplectic operators, defined as
\begin{eqnarray}
S(\xi,0,0) &=& \sum_{q=0}^{d-1}|\xi q\rangle\langle q| ;  \label{clifford} \\
S(1,\xi,0)&=&\sum_{q=0}^{d-1}\omega^{\xi q^2 2^{-1}}|q\rangle\langle q|;\label{clifford2} \\ 
S(1,0,\xi)&=&\sum_{q=0}^{d-1}\omega^{-\xi q^2 2^{-1}}|p_q\rangle\langle p_q|. \label{clifford3}
\end{eqnarray}
These operators are invertible and unitary  whenever the  values of $\xi$ and $d$ 
are  coprime (see proof of Lemma 2) and generate the normalizer of the Pauli group, 
which is usually refereed as the Clifford group.
The interested reader can check a more broad formulation in terms of a discrete phase-space in the Appendix A or in the Ref.~\cite{vourdas}.

\subsection{Qudit graph states}

We briefly review the theory of the so-called qudit graph states, which is well established in the literature \cite{quditgraph}.
The mathematical object used is a multi-graph $G=(V,E)$; 
we call $m_e\in\mathbb{Z}_d$ the multiplicity of the edge 
$e$, i.e., the number of times the edge appears.
Given a multigraph $G=(V,E)$, we associate a quantum state $|G\rangle$ in a $d$-dimensional system in the following way:
\begin{itemize}
 \item To each vertex $i\in V$ we associate a local state $|+\rangle= |p_0\rangle = d^{-1/2}\sum_{q=0}^{d-1} |q\rangle$. 
 \item For each edge $e=\{i,j\}$ and multiplicity $m_e$ we apply the unitary 
 \begin{eqnarray}
  Z^{m_e}_e= \sum_{q=0}^{d-1}|q_i\rangle\langle q_i|\otimes (Z_j^{m_e})^q
 \end{eqnarray}
 on the state $|+\rangle^V=\bigotimes_{i\in V}|+\rangle_i$. Thus, the graph state is defined as
 \begin{eqnarray}
  |G\rangle = \prod_{e\in E} Z^{m_e}_e |+\rangle^V.
 \end{eqnarray}
\end{itemize}
We allow among the edges $e\in E$ the presence of ``loops", i.e., 
an edge that contains only a single vertex. A loop of multiplicity 
$m$ on vertex $k$ means here that a local gate $(Z_k)^m$ is applied 
to the graph state. An example of a qudit graph state in a system 
with dimension $d>3$ is shown in Fig.~\ref{fig:graph}.

An equivalent way of defining a qudit graph state is via the stabilizer formalism \cite{quditgraph}. 
Given a multi-graph $G=(V,E)$, define for each vertex $i\in V$ the operator $K_i=X_i\prod_{e\in E}Z_{e\setminus\{i\}}$.  
The set $K_i$ generates an abelian group known as the stabilizer. 
The unique $+1$ common eigenstate of these operators is precisely the state $|G\rangle$ associated to the multi-graph $G$. Moreover, the set of common eigenstates of these operators form a basis of the global state space, the so-called graph state basis.  

The local action of the generalized Pauli group on a graph state is easy to picture and clearly preserves the graph state structure. As already said, the action of $Z^m_l$ corresponds to a loop of multiplicity $m$ on the qubit $l$, 
while the action of $X^m_l$ corresponds to loops of multiplicity $m$ on the qubits in the neighbourhood of the qubit $m$; 
this last observation is a corollary of Lemma 1.  

The action of the local Clifford group is richer and enables the conversion between different multi-graphs in a simple fashion. 
For prime dimensions, the action of the gate $S_k(\xi,0,0)$ enables the conversion between edges of different multiplicities, while the gate $S_k(1,-1,0)$ is associated to the operation known as local complementation - see Appendix B. Moreover, in non-prime dimensions the possible conversions between edges are governed by a greatest common divisor hierarchy, as show in more generality ahead - see Proposition 2 and Theorem 4.

\section{Qudit hypergraph states}

We introduce now the class of hypergraph states in a system with underlying finite dimension $d$. Before proceeding, we need first the concept of a controlled operation on a multipartite system.  
From a given local operation $M$ one can define a controlled operation $M_{ij}$ between qudits $i$ and $j$ as 
\begin{eqnarray}
M_{ij} = \sum_{q=0}^{d-1}|q_i\rangle\langle q_i|\otimes M^q_j   
\end{eqnarray}
Likewise, a controlled operation between three qudits $i$, $j$ and $k$ is defined recursively as 
\begin{eqnarray}
M_{ijk}= \sum_{q=0}^{d-1}|q_i \rangle\langle q_i |\otimes M^q_{jk} 
\end{eqnarray}
and in general the controlled operation between $n$ qudits labeled by $\mathcal{I}=\{i_1i_2\ldots i_n\}$ is given by
\begin{eqnarray}
M_{\mathcal{I}}=M_{i_1i_2\ldots i_n}= \sum_{q=0}^{d-1}|q_{i_1} \rangle\langle q_{i_1}|\otimes M^q_{i_2\ldots i_n} 
\end{eqnarray}
A prominent example is the CNOT operation, which is simply the bipartite controlled operation generated by the $X$ gate - $CNOT=\sum_q|q\rangle\langle q|\otimes X^q$. Although our formulation can be done in terms of this gate and its multipartite versions, it is preferable to use an equivalent formulation in terms of controlled-phase gates $Z_{\mathcal{I}}$, since these gates are mutually commuting and thus are easy to 
handle. Explicitly, the controlled phase gate on $n$ particles is given by
\begin{align}
Z_e &= \sum_{q_1=0}^{d} \dots  \!\!\! \sum_{q_{n-1}=0}^{d}
\ket{ q_1  \dots q_{n-1} }\bra{ q_1  \dots q_{n-1} }
Z^{q_1 \cdot \ldots \cdot q_{n-1}}
\nonumber
\\
&=\sum_{q_1=0}^{d} \dots  \! \sum_{q_{n}=0}^{d}
\ket{ q_1  \dots q_{n} }\bra{ q_1  \dots q_{n} }
\omega^{q_1 \cdot \ldots \cdot q_{n}}
\label{phasegate}
\end{align}

The mathematical object used here to represent a given state 
is a multi-hypergraph $G=(V,E)$;  as usual, we call $m_e\in\mathbb{Z}_d$ 
the multiplicity of the hyperedge $e$, i.e., the number of times the 
hyperedge appears. Given a multi-hypergraph $H=(V,E)$, we associate a 
quantum state $|H\rangle$ in a $d$-dimensional system in the following way:
\begin{itemize}
 \item To each vertex $i\in V$ we associate a local state $|+\rangle=d^{-1/2}\sum_{q=0}^{d-1} |q\rangle$. 
 \item For each hyperedge $e\in E$ with multiplicity $m_e$ we apply the controlled-unitary $Z^{m_e}_e$ on the state $|+\rangle^V=\bigotimes_{i\in V}|+\rangle_i$. 
 Thus, the hypergraph state is defined as
 \begin{eqnarray}
  |H\rangle = \prod_e Z^{m_e}_e |+\rangle^V.
 \end{eqnarray} 
\end{itemize}
We allow among the hyperedges $e\in E$ the presence of ``loops", i.e., an edge that contains only a single vertex. Also empty edges are allowed, they correspond to 
a global sign.
A loop of multiplicity $m$ on vertex $k$ means here that a local gate $(Z_k)^m$ is applied to the graph state. An example of hypergraph state is illustrated in 
Fig.~\ref{fig:hgraph}.

Equivalently, one can define a hypergraph state as the
unique $+1$ eigenstate of a maximal set of commuting 
stabilizer operators 
$K_i=X_i\prod_{e\in E}Z^{\dagger}_{e\setminus\{i\}}$; 
we explain this approach after the proof of Lemma 1 below. 

An important special class of hypergraph states are 
the so-called $n$-elementary hypergraph states, which are those constituted of a single hyperedge $e$ between $n$ qudits, i.e., the state has the simple form $|H\rangle=Z^{m_e}_e|+\rangle^V$. For this subclass, the main entanglement properties depend on the multiplicity $m_e$ of the hyperedge, as shown in the next sections. 

For completeness, we cite alternative formulations of hypergraph states that are potentially useful in other scenarios. First, we notice that the multiplicities of the hyperedges can be encoded as well in the adjacency tensor $\Gamma$ of the multi-hypergraph $H$, defined by $\Gamma_{i_1i_2\ldots i_n}=m_{\{i_1,i_2,\ldots,i_n\}}$, where $\{i_1,i_2,\ldots,i_n\}\in E$. 
For graph states, for example, the $\Gamma$ tensor is a matrix, the well-known adjacency matrix of the theory of graphs. Many local quantum operations, especially those coming from the local Clifford group, are elegantly described as  matrix operations over the adjacency matrix of the multi-graph $G$. 

One can also work in the Schr\"odinger picture: the form of the state in the computational basis is given by:
\begin{eqnarray}
|H\rangle=\sum_{\mathbf{q}=\mathbf{0}}^{\mathbf{d-1}}\omega^{f(\mathbf{q})}|\mathbf{q}\rangle,
\end{eqnarray}
where $\mathbf{q}\equiv (q_1,q_2,\dots,q_n)$ 
and $f$ is a function from $\mathbb{Z}_d$ to $\mathbb{Z}_d$ defined by 
$f(\mathbf{q})=\bigoplus_{\{q_1,\ldots,q_k\}\in E} \bigwedge_k q_k$. For qubits, for example, the function $f$ is a Boolean function and this encoding is behind applications such as Deutsch-Jozsa and Grover's algorithms \cite{hypergraphpapers}.

\section{Properties of hypergraph states and the stabilizer formalism}

In the following sections, we derive some properties of hypergraph states, the 
controlled $Z$ operation on many qudits and the stabilizer formalism. These tools
will later be used for the SLOCC and LU classification.

\subsection{Local action of Pauli and Clifford groups}

We now consider the effect of unitaries from the Pauli and Clifford groups 
on a hypergraph state. 
First, we need some simple relations:
\begin{lemma}
\label{lemmacommutation}
The following relations hold:
\begin{eqnarray}
X^{\dagger}_kZ_{\mathcal{I}}X_k=Z_{\mathcal{I}\setminus \{k\}}Z_{\mathcal{I}} \label{relation1}
\end{eqnarray}
\begin{eqnarray}
Z_{\mathcal{I}}X_kZ^{\dagger}_{\mathcal{I}}=X_kZ_{\mathcal{I}\setminus \{k\}} \label{relation2}
\end{eqnarray}

\end{lemma}
\begin{proof} We prove by induction on the cardinality $n$ of the index set $\mathcal{I}$, i.e., on the number of qudits. For $n=2$, remembering the relation $X^{\dagger}ZX=\omega Z$, we see that
\begin{eqnarray*}
X^{\dagger}_kZ_{jk}X_k &=& X^{\dagger}_k\left(\sum_{q=0}^{d-1}|q_j\rangle\langle q_j|\otimes Z^q_k\right)X_k\\
&=&\sum_{q=0}^{d-1}|q_j\rangle\langle q_j|\otimes (X^{\dagger}_kZ^q_kX_k)\\
&=&\sum_{q=0}^{d-1}|q_j\rangle\langle q_j|\otimes (\omega^qZ^q_k)\\
&=&\left(\sum_{q=0}^{d-1}\omega^q|q_j\rangle\langle q_j|\otimes I_k\right)\left(\sum_{q=0}^{d-1}|q_j\rangle\langle q_j|\otimes Z^q_k\right)\\
&=& Z_jZ_{jk}
\end{eqnarray*}
and the relations are valid. Now, let us consider the set $\mathcal{I}$ having cardinality $n$ and the set $\mathcal{I}'=\mathcal{I}\cup\{j\}$, with $j\neq k$. Then we have
\begin{eqnarray*}
X^{\dagger}_kZ_{\mathcal{I}'}X_k &=& X^{\dagger}_k\left(\sum_{q=0}^{d-1}|q_j\rangle\langle q_j|\otimes Z^q_{\mathcal{I}}\right)X_k \\
&=&\sum_{q=0}^{d-1}|q_j\rangle\langle q_j|\otimes (X^{\dagger}_kZ_{\mathcal{I}}^qX_k)\\
&=&\sum_{q=0}^{d-1}|q_j\rangle\langle q_j|\otimes (Z_{\mathcal{I}\setminus \{k\}}Z_{\mathcal{I}})^q\\
&=& \left(\sum_{q=0}^{d-1}|q_j\rangle\langle q_j|\otimes Z_{\mathcal{I}\setminus \{k\}}^q\right)\left(\sum_{q=0}^{d-1}|q_j\rangle\langle q_j|\otimes Z_{\mathcal{I}}^q\right)\\
&=& Z_{\mathcal{I}'\setminus\{k\}}Z_{\mathcal{I}'}
\end{eqnarray*}
i.e., if the relations are valid for $n$, then they are also valid for $n+1$.
\end{proof}

The effect of applying the gate $X^{\dagger}_k$ on an elementary hypergraph 
is then to create a hyperedge of same multiplicity on the neighbourhood of 
the qudit $k$:
\begin{eqnarray*}
X^{\dagger}_k|H\rangle = X^{\dagger}_kZ^{m_e}_e|+\rangle^V = Z^{m_e}_{e\setminus \{k\}}Z^{m_e}_e X^{\dagger}_k|+\rangle^V= 
Z^{m_e}_{e\setminus \{k\}}|H\rangle
\end{eqnarray*}
It is important to notice that the hyperedges induced on the neighbourhood 
of $k$ by repeated applications of $X_k^{\dagger}$ have multiplicities that 
are divisible by 
$m_e$. Depending on the primality of the underlying dimension $d$, there 
are restricted possibilities of inducing hyperedges on neighbouring qudits via this procedure, a difference in relation to the qubit case. 

Moreover, since the local $Z_k$ gate commutes with any $Z_{\mathcal{I}}$, it can 
always be locally removed by applying $Z_k^{\dagger}$. 
As explained previously, we adopt the convention of representing any $Z_k$ acting on a hypergraph state as a loop - a $1$-hyperedge - around the vertex $k$; 
higher potencies $(Z_k)^m$ are represented by $m$ loops around $k$. 
Thus, the action of the local Pauli group constituted of the unitaries $X_k^mZ_k^n$ 
is to create $n$ loops on the vertex $k$ and $m$-hyperedges on the neighbourhood of $k$.    

Let us turn now to the action of the local Clifford group. The local
gate from Eq.~(\ref{clifford}) performs permutations on the computational 
basis, via the mapping $S(\xi^{-1},0,0)ZS^{\dagger}(\xi^{-1},0,0)=Z^{\xi}$, 
where $Z$ acts on a single particle. Based on this, we can derive the action
on multiparticle $Z_e$ gates, corresponding to a hyperedge $e$. It turns out 
that for $d$ prime, it is always possible to convert a $k$-hyperedge of multiplicity 
$m$ ($m \neq 0 \mod d$) to another $k$-hyperedge of multiplicity $m'$ ($m' \neq 0 \mod d$). For non-prime 
$d$, the $k$-hyperedges that are connectable via permutations are those whose multiplicities have an equal greatest common divisor with the dimension $d$. In 
detail, we can formulate

\begin{proposition}
\label{propclifford}
Let $k,k'\in \mathbb{Z}_d$ be such that $gcd(d,k)=gcd(d,k')=g$. Then there exists 
a Clifford operator $S$ such that $S (Z_e)^k S^{\dagger}=(Z_e)^{k'}$.
\end{proposition}
\begin{proof}
One can find the proof of this proposition in Ref.~\cite{cliff1}. 
\end{proof}

For a detailed discussion on the Clifford group see
Ref.~\cite{cliff2}.  For completeness, we provide in Appendix C
alternative proofs of this Proposition. The local Clifford gates in Eq.~(\ref{clifford2}) and Eq.~(\ref{clifford3}) 
are associated with the local complementation of qudit graphs, which is 
explained in detail in Appendix B. 

\subsection{Stabilizer formalism}

From relation (\ref{relation2}) we can construct the stabilizer
operator on a vertex $i$:
\begin{eqnarray}
|H\rangle &=& \prod_{e\in E}Z_e|+\rangle^V = \prod_{e\in E}Z_eX_iZ^{\dagger}_eZ_e|+\rangle^V \\ &=& X_i\prod_{e\in E^*}Z_{e\setminus\{i\}}|H\rangle=K_i|H\rangle
\end{eqnarray}
with $K_i=X_i\prod_{e\in E*}Z_{e\setminus\{i\}}$ and where $E^*$ denotes all edges containing $i$. Hence, the operators $K_i$ stabilize the hypergraph state $|H\rangle$. An equivalent way is expressing the stabilizer operator in the compact way $K_i=X_iZ_{N_i}$, where $Z_{N_i}\equiv\prod_{j\in N_i}Z_j$, where $N_i$ is the neighbourhood of $i$. Moreover, these operators are mutually commuting:
\begin{eqnarray}
K_iK_j &=& (X_i\prod_{e\in E}Z_{e\setminus\{i\}})(X_j\prod_{e\in E}Z_{e\setminus\{j\}}) \\
&=& (\prod_{e\in E}Z_eX_i\prod_{e\in E}Z^{\dagger}_e)(\prod_{e\in E}Z_eX_j\prod_{e\in E}Z^{\dagger}_e) \\
&=& \prod_{e\in E}Z_eX_iX_j\prod_{e\in E}Z^{\dagger}_e \\
&=& \prod_{e\in E}Z_eX_jX_i\prod_{e\in E}Z^{\dagger}_e = K_jK_i
\end{eqnarray}
Indeed, these operators generate a maximal abelian group on the number $n$ of qudits. The group properties of closure and associativity are straightforward, while the identity element comes from $K_i^{d}=I$ and the inverse of $K_i$ being simply $K_i^{\dagger}$. Each operator in this group has eigenvalues $1, \omega, \omega^2, \ldots,\omega^{d-1}$ and their $d^n$ common eigenvectors form an orthonormal basis of the total Hilbert space, the hypergraph basis with elements given by
\begin{eqnarray}
|H_{k_1,k_2,\ldots,k_N}\rangle=Z_1^{-k_1}Z_2^{-k_2}\ldots Z_N^{-k_N}|H\rangle
\end{eqnarray}
where the $k_i$s assume values in $\mathbb{Z}_d$. Notice also that
\begin{eqnarray}
|H\rangle\langle H|=\frac{1}{d^N}\prod_{i\in V}(I+K_i+K^2_i+\ldots+K^{d-1}_i)
\end{eqnarray}
\newline
In the qubit case, these non-local stabilizers are observables and were used for the development of novel non-contextuality and locality inequalities \cite{hyper1, hyper2}. For $d>2$, these operators are no longer self-adjoint in general, but we believe techniques similar to Ref.~\cite{cps} could be used to extend the results of the qubit 
case to arbitrary dimensions.  

\subsection{Local measurements in $Z$ basis and ranks of the reduced states}

It is possible to give a graphical description of the measurement of a non-degenerate observable $M=\sum_q m_q|q\rangle\langle q|$ on a hypergraph state in terms of hypergraph operations. 
Obtaining outcome $m_q$ when measuring $M$ on the qudit $k$ of the hypergraph state $|H\rangle=\prod_e Z_e^{m_e}|+\rangle^V$ amounts to performing the projection 
$P^{(k)}_q|H\rangle$, where $P^{(k)}_q=|q_k\rangle\langle q_k|$. 
The state after the measurement is then $|q_k\rangle\otimes|H'\rangle$, where $|H'\rangle = \prod_{e}Z_{e\setminus\{k\}}^{m_e \cdot q}|+\rangle^{V\setminus\{k\}}$.

\begin{figure}
\begin{center}
\includegraphics[width=6cm]{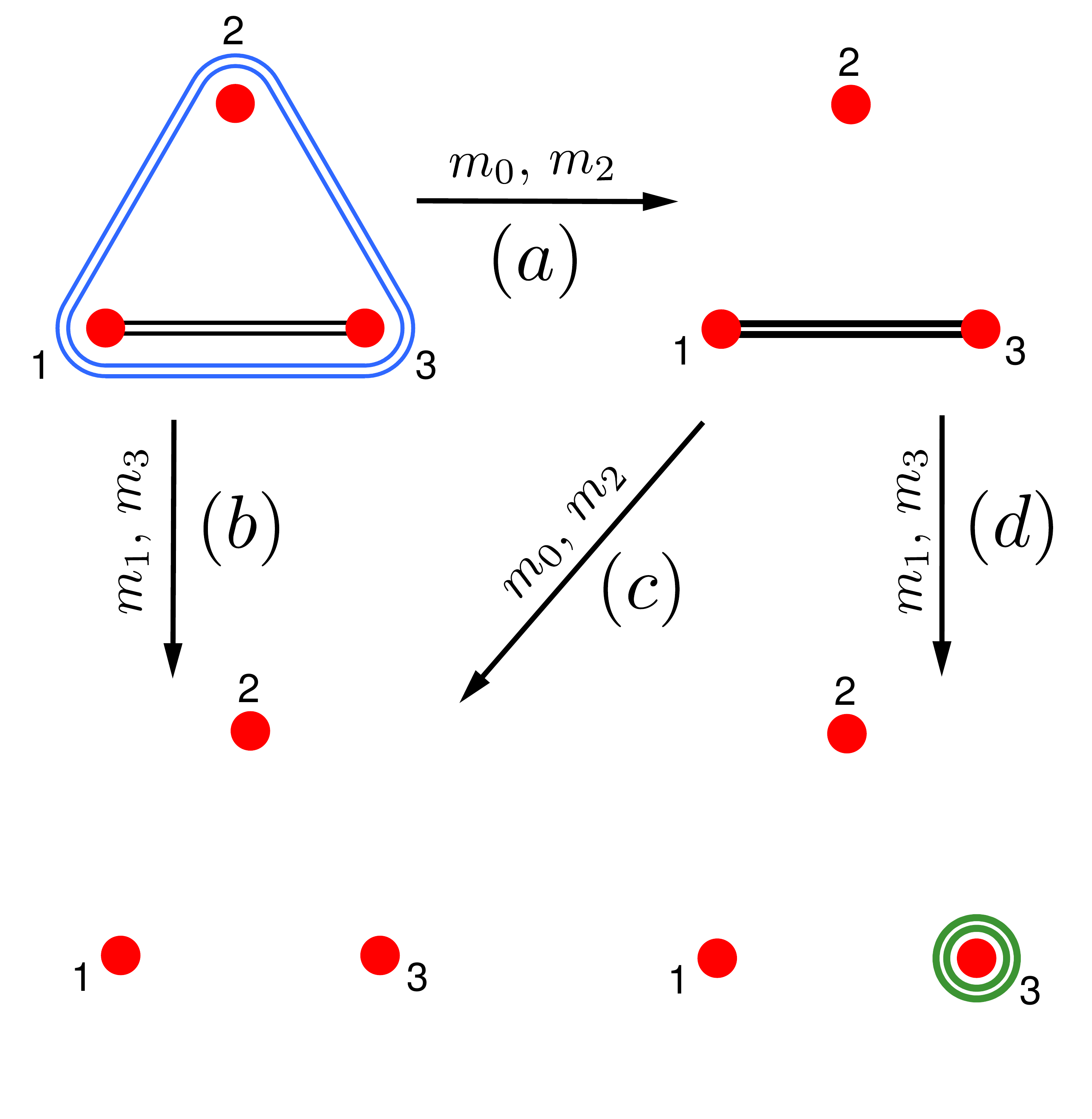} 
\end{center}
\vspace{-.5cm}
\caption{Measurements in $Z$-basis on the hypergraph state $|H\rangle=Z_{123}^2Z_{13}^2|+\rangle^V$ ($d=4$). Measurement on qudit $2$ results in (a) for outcomes $m_0$ or $m_2$ and (b) for outcomes $m_1$ or $m_3$. Now, measuring on qudit $1$ results in (c) for outcomes $m_0$ or $m_2$ and (d) for outcomes $m_1$ or $m_3$. It is clear that any reduced state has rank $2$.}
\label{fig:scheme}
\end{figure}

Moreover, the calculation of ranks of reduced states can be done graphically. 
We show now a lemma that will be important in future derivations:
\begin{lemma}
\label{lemmarank}
For a $d$-dimensional system, the rank of any reduced state of an $n$-elementary 
hypergraph state is $d/gcd(d,m_e)$ where $m_e$ is the multiplicity of the 
hyperedge.  
\end{lemma}
\begin{proof} An $n$-elementary hypergraph state is given by
\begin{eqnarray}
|H\rangle = Z_e^{m_e}|+\rangle^V = \sum_{q=0}^{d-1}|q\rangle\langle q|
\otimes (Z_{e\setminus\{1\}}^{m_e})^q|+\rangle^V,
\end{eqnarray}  
where $e=\{12\ldots n\}$ is the $n$-hyperedge. Let us first consider the case
where only a single system is traced out. Tracing out subsystem $1$, we arrive at
\begin{eqnarray}
\rho_{e\setminus\{1\}} 
&=& Tr_1(|H\rangle\langle H|) = \sum_{q=0}^{d-1}\langle q|H\rangle\langle H|q\rangle \\
                       &=&\frac{1}{d}\sum_{q=0}^{d-1} |H^{(1)}_q\rangle\langle H^{(1)}_q|
\end{eqnarray}
where 
\begin{eqnarray}
|H^{(1)}_q\rangle =(Z_{e\setminus\{1\}})^{qm_e}|+\rangle^{V\setminus\{1\}}.
\end{eqnarray}
The number of different values (modulo $d$) of the product $qm_e$, with
$q=0,1,2,\ldots, d-1$ is $d/gcd(d,m_e)$, since in $\mathbb{Z}_d$ one 
has $m_eq=m_eq'$ iff $q=q'$ modulo $d/gcd(d,m_e)$. 

Moreover, the $d/gcd(d,m_e)$ different vectors $|H^{(1)}_q\rangle$ are 
linearly independent. We prove this via induction on the number of 
vertices of $V\setminus\{1\}$. For $V\setminus\{1\}$ composed of one 
vertex, we see that $|H^{(1)}_q\rangle =Z_2^{qm_e}|+_2\rangle$ and 
\begin{eqnarray}
\sum_q \alpha_q |H^{(1)}_q\rangle = (\sum_q\alpha_qZ_2^{qm_e})|+\rangle_2
\end{eqnarray}
where $\alpha_q$ are arbitrary complex numbers. Then 
$\sum_q \alpha_q |H^{(1)}_q\rangle=0$ iff $(\sum_q\alpha_q Z_2^{qm_e})=0$. 
To see this, not that the operators $Z_2^{qm_e}$ are diagonal in the
computational basis and the vector $|+\rangle_2$ is an equal superposition
of all basis elements. Therefore, $(\sum_q\alpha_q Z_2^{qm_e})|+\rangle_2=0$.
implies already $(\sum_q\alpha_q Z_2^{qm_e})=0$. 

Since the Pauli operators form a basis of the Lie algebra $sl_d(\mathbb{C})$ and are 
thus linearly independent and given that the operators $Z^{qm_e}$ are Pauli 
operators we have that $(\sum_q\alpha_qZ_2^{qm_e})=0$ implies that $\alpha_q=0$ 
for all $q$, i.e., the vectors $|H^{(1)}_q\rangle$ are linearly independent. 

For $V\setminus\{1\}$ composed of two vertices, $|H^{(1)}_q\rangle =Z_{23}^{qm_e}|+,+\rangle_{2,3}$. By the same argumentation above, $|H^{(1)}_q\rangle$ are linearly independent iff $Z_{23}^{qm_e}$ are linearly independent operators. We have
\begin{eqnarray}
\sum_q\alpha_qZ_{23}^{qm_e}&=&\sum_{q,q'}|q'_2\rangle\langle q'_2|\otimes(\alpha_qZ_3^{q'qm_e}) \\
&=&\sum_q'|q'_2\rangle\langle q'_2|\otimes(\sum_q\alpha_qZ_3^{q'qm_e})\\
&=&|0_2\rangle\langle 0_2|\otimes(\sum_q\alpha_q I_3) \\
&+& |1_2\rangle\langle 1_2|\otimes(\sum_q\alpha_q Z^q_3) + \ldots
\end{eqnarray}
We see then that $\sum_q\alpha_qZ_{23}^{qm_e}=0$ can be satisfied only 
if $\alpha_q=0$ for all $q$, since this is the only way to have a null 
term $|1_2\rangle\langle 1_2|\otimes(\sum_q\alpha_q Z^q_3)$, given that 
the $Z_3^q$ operators are linearly independent. Thus $|H^{(1)}_q\rangle$ 
are linearly independent for  $V\setminus\{1\}$ composed of two vertices 
as well. The general induction argument is now clear and it is obvious 
that the vectors $|H^{(1)}_q\rangle$ are linearly independent in general.  
Hence the rank of $\rho_{e\setminus\{1\}}$ is $d/gcd(d,m_e)$. 

One can directly check using the representation in Eq.~(\ref{phasegate})
that if more than two qudits are traced out the same arguments apply.
Thus, the rank of any reduced state is $d/gcd(d,m_e)$. \end{proof}   

\section{SLOCC and LU classes of hypergraphs}

\subsection{SLOCC and LU transformations}

The phenomenon of entanglement is a consequence of the physical restriction 
to local operations by agents separated by space-like distances. It is thus 
important to identify when is it possible to inter-convert two different 
quantum states by means of local operations or, more specifically, characterize
their equivalence under SLOCC or LU operations. Finding the SLOCC/LU classes to 
which a given hypergraph state belong is in general a cumbersome task even in 
the qubit case, but in the following we will present several results and ideas
that can be used for tackling this task. 

Let first define the basic notation. Two pure $n$-partite states $|\phi\rangle$ and 
$|\psi\rangle$ are equivalent under local unitaries  if one has a relation like
\begin{eqnarray}
|\phi\rangle = \bigotimes_{i=1}^{n} U_i|\psi\rangle, 
\end{eqnarray}
where the $U_i$ are unitary matrices, acting on the $i$-th particle. The question
whether two multiqubit states are LU equivalent or not can  be decided by bringing the states into a normal form under LU transformations \cite{kraus}.

More generally, the states are equivalent 
under stochastic local operations and classical communication 
(SLOCC) iff there exist invertible local operators (ILOs) $A_i$ such that
\begin{eqnarray}
|\phi\rangle = \bigotimes_{i=1}^{n} A_i|\psi\rangle.
\end{eqnarray}
Physically, this means that $\ket{\phi}$ can be reached starting from 
$\ket{\psi}$ by local operations and classical communication with a non-zero
probabilitiy. 

Although general criteria for SLOCC-equivalence of multiparite states do 
not exist, it is possible to find general necessary conditions that are
useful as exclusive constraints. For instance, SLOCC transformations can clearly
not change the rank of a reduced state $\rho_i$. Moreover, for special classes 
of states sufficient conditions for SLOCC-equivalence can be found.

\subsection{Elementary hypergraphs}

We now address the problem of SLOCC classification of $n$-elementary hypergraph
states. The classification depends on the greatest common divisor between the 
underlying dimension and the hyperedge multiplicity, as show in the following 
theorem:

\begin{theorem}
\label{theoremelementary}
For a $d$-dimensional $n$-partite system, two $n$-elementary hypergraph states 
with hyperedge multiplicities $k$ and $k'$ are equivalent under LU, and hence 
also under SLOCC if $gcd(d,k)=gcd(d,k')$. For the case that $gcd(d,k)\neq gcd(d,k')$
the states are inequivalent under SLOCC.
\end{theorem}

\begin{proof} 
If $gcd(d,k)=gcd(d,k')$ we can, according to Proposition 1 find a local Clifford
transformation with $S Z^k S^\dagger = Z^{k'}.$ So we have 
$S \ket{H} = S Z^k \ket{+}^n = Z^{k'} S \ket{+}^n = Z^{k'} \ket{+}^n = \ket{H'}$
since $S\ket{+}^n= \ket{+}^n.$ For the other implication, note that
if $gcd(d,k)\neq gcd(d,k')$, the single-system reduced states have different 
ranks by Lemma \ref{lemmarank} and thus the states are not SLOCC equivalent. 
\end{proof}

In other words, the number of different elementary hypergraph SLOCC classes 
is the number of different values (modulo $d$) of $gcd(d,k)$, which is 
obviously the number of divisors of $d$. It is remarkable that in this 
case SLOCC equivalence is the same as equivalence under Local Clifford 
operations, by Proposition \ref{propclifford}.  For $d$ prime, all values $k\in \mathbb{Z}_d$ are obviously coprime with $d$ and hence the following implication is straightforward:

\begin{corollary}
\label{corrprime}
For $d$ prime, all $n$-elementary hypergraph states are equivalent under SLOCC. 
\end{corollary}

Finally, we note that also some other hypergraphs states are LU equivalent to
elementary hypergraph states. From Lemma \ref{lemmacommutation} 
and the discussion that followed, one 
sees that the action of the local gate $X_i^{\dagger}$ on an $n$-elementary 
hypergraph state creates a $n-1$-hyperedge on the neighbourhood of $i$ with 
equal multiplicity $m_e$ of the $n$-hyperedge $e$. Acting $k$ times with 
this local gate, i.e., application of $(X_i^{\dagger})^k$ results in inducing, 
in the neighbourhood of the qudit $i$, a $n-1$-hyperedge of multiplicity 
$km_e$. As shown in the proof of Lemma \ref{lemmarank}, the number of different values 
of the product $km_e$ is given by $d/gcd(d,m_e)$. Thus, the higher the value $gcd(d,m_e)$, the smaller the possible $n-1$-hyperedges that can be created 
(or erased) within an elementary hypergraph state.    

\subsection{Tools for SLOCC classification}

In this section we will explain some more refined criteria for proving 
or disproving SLOCC equivalence. As already mentioned, the rank of the 
reduced states, $r(\rho_i)=Tr_{S \backslash i}(\rho_S)$ (where S denotes 
the set of all subsystems), as a simple way of identifying inequivalences. 

To find a finer distinction we employ a method based on Ref.~\cite{Lamata} 
that uses a $1|23\ldots n$ split of the system to identify types of
inequivalent bases of the ($2,3, \ldots ,n$)-subspace, which results in a 
lower bound on the number of actual SLOCC-classes. As we want to  infer for 
a given state its SLOCC-class, there remains the following problem to be solved: identifying the basis 
which has minimal entanglement in its basis vectors. Accordingly, we refer 
to this tool as minimally entangled basis (MEB) criterion.  A major disadvantage 
of this method is that with growing number $n$ of subsystems, the entanglement 
structure within the bases becomes more complex, as it rises recursively from 
the total number of SLOCC-classes of the ($n-1$)-partite systems. 

The MEB of an $n$-partite quantum state is defined as follows:
\begin{definition}
Consider a state 
\begin{eqnarray*}
\ket{\psi_{12 \ldots n}}=\sum_{a_1, a_2,\ldots, a_n=0}^{d-1,d-1,...,d-1}
c_{a_1,a_2,\ldots, a_n}\ket{a_1,a_2,\ldots, a_n}
\end{eqnarray*}
in a $d^{\times n}$ system. According to Ref. \cite{Lamata}, we define the $d \times (d^{n-1})$
coefficient-matrix $C_{1|2\ldots n}$ in the canonical basis $\{e_i\}$ as follows:
\begin{eqnarray*}
C_{1|2\ldots n}= \sum_{a_1, a_2,\ldots, a_n=0}^{d-1,d-1,...,d-1} c_{a_1,a_2,\ldots, a_n} e_1(e_2^T\otimes\ldots\otimes e_n^T)
\end{eqnarray*}
This matrix holds all information about the total state. From the singular value decomposition (SVD) of this matrix, $C_{1|2\ldots n}=U_1 D V_{2\ldots n}^{\dagger}$, we can identify a basis $\{v_k\}$ of the right subspace $(2,\ldots, n)$, where individual basis vectors $\ket{v_k}$ might be entangled.\newline
Within this framework, we define a minimally entangled basis (MEB) $\{v_k\}_{MEB}$ of $\ket{\psi_{12 \ldots n}}$ as the one within which the number of full product vectors is maximal under the condition that it spans the same subspace as $\{v_k\}$. 
\end{definition}

With this definition we can state:

\begin{lemma}
 Two $n$-partite quantum states $\ket{\phi}, \ket{\psi}$ of the
 same subsystem-dimensionality and equal reduced ranks are 
 SLOCC-inequivalent, if their MEBs have a different number 
 of product vectors. 
\label{l:slocc-ineq}
\end{lemma}
\begin{proof}
The action of the ILOs $A_i$, where $i=1,2,\ldots n$, on $C_{1|2\ldots n}$ 
in its SVD is identified to be 
\begin{equation}
A_1 U_1 D [(A_2 \otimes\ldots A_n) V_{2 \ldots n}]^{\dagger}.
\end{equation} 
We analyse the basis $\{v_k\}$ of the right subspace. 
The Schmidt rank of each basis vector can be changed by $A_1$ exclusively,
which corresponds to a basis transformation of the subspace. 
If the states $\ket{\phi}, \ket{\psi}$ are SLOCC equivalent, by 
definition there exists ILOs $A_i$ which map $\ket{\phi}$ into 
$\ket{\psi}$ and thus map the basis of the right subspace of 
$\ket{\phi}$ into the basis of the right subspace of $\ket{\psi}$. 
The MEB of $\ket{\psi}$ will be then a valid MEB for $\ket{\phi}$, implying that 
the number of product vectors is the same.
\end{proof}

In the above Lemma we consider states $\ket{\phi}, \ket{\psi}$ that have 
equal reduced ranks, because otherwise these states are automatically
SLOCC-inequivalent and there is no point in calculating their MEBs.

Notice that inequivalent MEBs can exclude SLOCC equivalence, but an equivalence of MEBs does not, in general, guarantee SLOCC-equivalence. An exception is the case where the right subspace is spanned by a complete product basis. The reason is that in this case
they are SLOCC equivalent to a generalized GHZ state:

\begin{lemma}
 Two genuine $n$-partite entangled quantum states $\ket{\phi}$, $\ket{\psi}$ of 
 the same subsystem-dimensionality and equal reduced single-particle ranks are
 SLOCC-equivalent, if their MEBs are complete product bases.
 \label{l:slocc-eq}
\end{lemma}
\begin{proof}  

 We show that the existence of a complete product basis within the right 
 subspace is sufficient to find ILOs that transform $\ket{\psi}$ (and $\ket{\phi}$)
 to the GHZ type state $\ket{\psi_{GHZ}} \sim \sum_{k=0}^{r} \bigotimes_{i=1}^n  \ket{k}_i$, where $r$ is the rank of the reduced single-particle states.
 
 Let us assume that the basis vectors $\ket{v_k}$ are all product vectors.
 Therefore, they can be written as 
 \begin{equation}
  \ket{v_k} = \bigotimes_{i=2}^n \ket{\phi^{(k)}_i}
 \end{equation}
 In order to map this onto the GHZ state, we only have to find ILO $A^{(i)}$
 on the particles $i=2,\ldots,n$ such that for any particle the set of states
 $\{ \ket{\phi^{(k)}_i}\}$ is mapped onto the  states $\{\ket{k}_i\}$. This is 
 clearly possible: since the reduced state ranks are all $r$, the set
 $\{ \ket{\phi^{(k)}_i}\}$ consists of $r$ linearly independent vectors. 
 Finally, on the first particle we have to consider the left basis $\ket{u_k}$.
 These vectors are orthogonal, and we can find a unitary transformation 
 that maps it to $\{\ket{k}_1\}$.
\end{proof}

Based on the Lemmata presented in this subsection we wrote computer programs which we regard as tools which we use later for classification of tripartite hypergraphs of dimension $3$ and $4$. 
\subsubsection{Tool \# 1}
The first program checks whether there exist a state $\varrho$ in the subspace spanned
by a given set of pure states $|v_i\rangle$ for $i=1,\dots,K,\; K\leq d$ that has a positive partial transpose (PPT) with respect to any bipartition \cite{ppt}.   
This problem can be formulated as following semidefinite program (SDP)
\begin{eqnarray}
\min_\lambda &\; & 0\\
\text{subject to} &\; & \varrho = \sum_{ij}^K\lambda_{ij}|v_i\rangle\langle v_j|,\nonumber\\
&\; & \varrho \geq 0,\nonumber\\
&\; & \forall\; \text{bipartitions}\; M|\overline{M},\; \varrho^{T_M} \geq 0,\nonumber\\
&\; & \lambda^\dagger = \lambda,\; \Tr(\lambda) = 1,\nonumber
\end{eqnarray}
where last condition means $\lambda$ is a hermitian $K\times K$ matrix with trace 1, and $\varrho^{T_M}$ denotes partial transpose of matrix $\varrho$ with respect to the subsystem $M$. 

This tool can be used to prove that there is no product vector in the right subspace $(2,\ldots,n)$ of an $n$-partite state $|\phi\rangle$, where $K$ is the number of basis vectors in the right subspace. If the above SDP is infeasible it implies that there is no separable state in the subspace $(2,\ldots,n)$, which in turns implies that there is no product vector. If for some other $n$-partite state $\ket{\psi}$ there is a product vector in the right subspace $(2,\ldots,n)$ the two states $\ket{\phi}$ and $\ket{\psi}$ are SLOCC-inequivalent according to Lemma~\ref{l:slocc-ineq}.
\subsubsection{Tool \# 2}
The second program is a slight modification of the first one and it checks 
whether there exist a PPT state of rank $K$ in the subspace spanned by $K$ linearly independent vectors $|v_i\rangle,\;i=1,\dots,K,\; K\leq d$.
If the optimal value $\epsilon$ of the following semidefinite program
\begin{eqnarray}
\min_{\lambda,\epsilon} &\; & \epsilon\\
\text{subject to} &\; & \varrho = \sum_{ij}\lambda_{ij}|v_i\rangle\langle v_j|,\nonumber\\
&\; & \varrho \geq 0,\nonumber\\
&\; & \forall\; \text{bipartitions}\; M|\overline{M},\; \varrho^{T_M} \geq 0,\nonumber\\
&\; & \varrho \geq \epsilon \left(\sum_i |v_i\rangle\langle v_i|\right),\nonumber\\
&\; & \lambda^\dagger = \lambda,\; \Tr(\lambda) = 1,\nonumber
\end{eqnarray}
is greater than 0, and if the found PPT state $\varrho$ can be proven to be (fully) separable,
then by the range criterion (see Ref.~\cite{range}) it means that in the subspace spanned by $|v_i\rangle$ there are $K$ product states which span the same subspace. This program can be used to prove SLOCC-equivalence
of states $\ket{\phi}$ and $\ket{\psi}$ according to Lemma~\ref{l:slocc-eq}, if for both states 
the above conditions are satisfied for their right subspace of at least one bipartition.
\subsubsection{Tool \# 3}
Finally, it is convenient to perform a numerical optimization in order to find product states $|v^p_i\rangle$, $i = 1,\dots,K',\;K'\leq K$ in the subspace spanned by the given set of vectors $|v_i\rangle$ for $i=1,\dots,K,\;
K\leq d$. This can be done by, let us say, maximizing the purity of the reduced states [that is, $1-\Tr(\varrho^2_M)$, 
where $\varrho_M = \Tr_M(\varrho)$ is the reduced state of the subsystem $M$] for each bipartition and
minimizing the scalar product $\left| \langle v^p_i|v^p_j\rangle \right|^2$ between each pair of product vectors for each unique pair $\{i,j\},\;i,j \in \{1,\dots, K'\}$. Minimizing the scalar products makes the program to look for linearly independent vectors which in the best case are orthogonal. 

As we will see in the next section, for 
most of the tripartite hypergraph states of dimension $3$ and $4$ numerical optimization, described above, 
gives explicit form of product states in the right subspace if they exist. Moreover, knowing the exact form of 
product states for the case where a full product basis exists for both states $\ket{\phi}$ and $\ket{\psi}$ allows us to find an explicit SLOCC transformation between these states.
\subsection{Tools for LU-classification}

Let us now discuss tools how LU equivalence can be characterized. In principle, 
this question can be decided using the methods of Ref.~\cite{kraus}, but for the
examples in the next section some other methods turn out to be useful. 

If LU equivalence should be proven, an obvious possibility is to find directly
the corresponding LU tranformation. This has been used in Theorem \ref{theoremelementary}. For proving non-equivalence, one can use 
entanglement measures such as geometric measure \cite{GeomEnt}, since
such measures are invariant under LU transformations. Another possibility
is the white-noise tolerance of witnessing entanglement \cite{PPTMixer}. The latter
method works as follows:
For an entangled state which is detected by a witness one can assign an upper limit of white noise which can be added to the state, such that the state can still be detected by that witness. Clearly, if two states are equivalent under LU, they have 
the same level of white-noise tolerance of entanglement detection. Now if
one considers a class of decomposable witnesses the estimation of this level 
for a given state can be accomplished effectively by means of semidefinite 
programming \cite{SDP}. Below, we use a method described in 
Ref.~\cite{PPTMixer} to witness genuine multipartite entanglement of 
hypergraph states and to determine the corresponding white-noise 
tolerance of that witnessing. 
 
Using the tools described above we present a classification in terms of SLOCC- and LU-equivalence of tripartite hypergraph states for dimensions $3$ and $4$ in the next 
section.

\section{Examples}
\afterpage{%
\begin{center}
\vspace{-11pt}
\begin{table*}[p]
\begin{tabularx}{\textwidth}{| c | c | c | c |}
\hline
\parbox[c][1.5cm]{.018\textwidth}{\rotatebox{90}{Class}} & \parbox[c]{.07\textwidth}{Schmidt ranks} & \parbox[c][1.5cm]{.8\textwidth}{Representatives} &  \parbox[c][1.5cm]{.077\textwidth}{Geom. measure/ w-noise tolerance } \\ \hline
\multirow{ 2}{*}{\bf 1} & \parbox[c]{.07\textwidth}{\vspace{12pt} 1$|$23\quad 3\\ 2$|$13\quad 3\\ 3$|$12\quad 3} & \parbox[c][2cm]{.8\textwidth}{
\includegraphics[width=1.8cm]{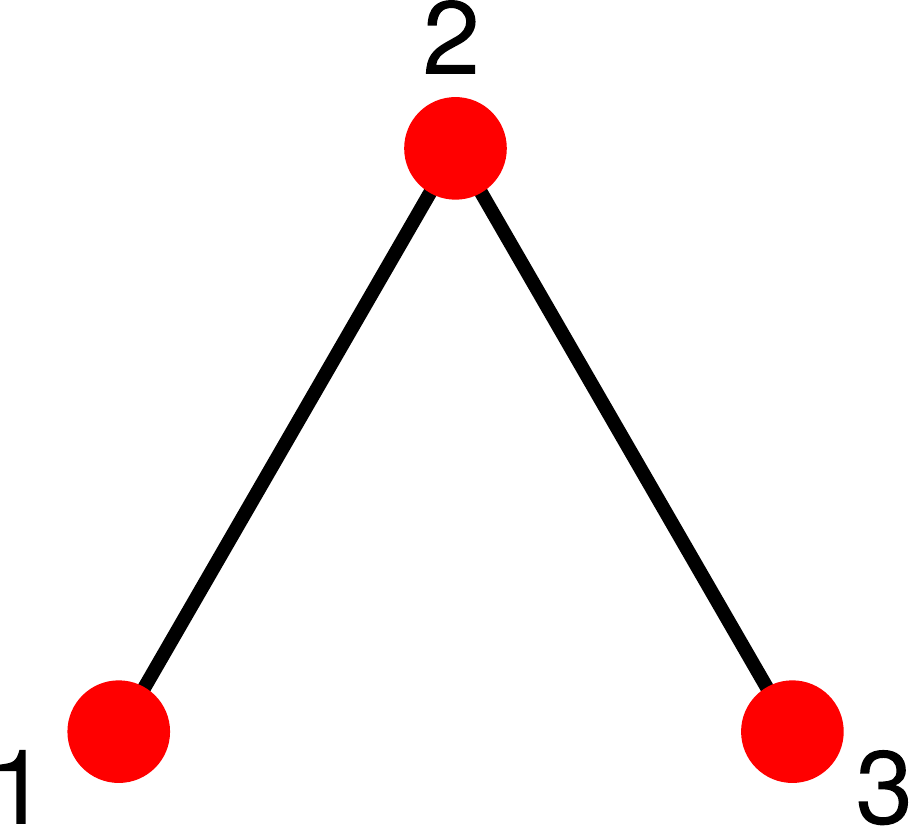} \quad\; \includegraphics[width=1.8cm]{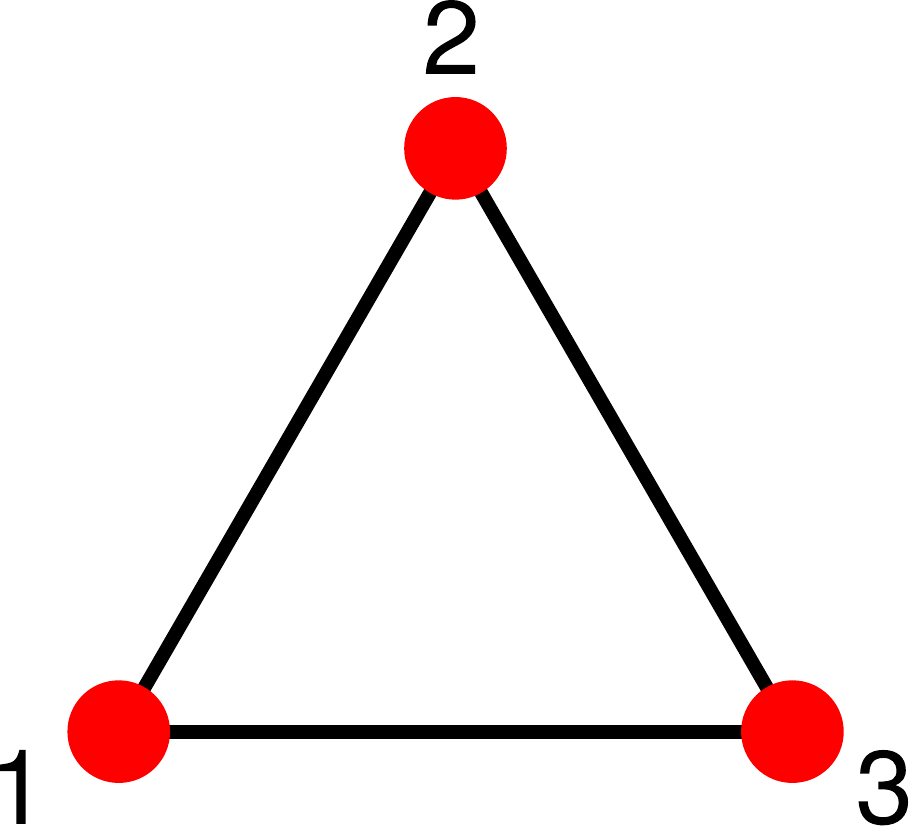}
} & \parbox[c][1.5cm]{.07\textwidth}{$\sim$ 0.66 \\ 62.5\%} \\ \cline{3-4}
& & \parbox[c][2cm]{.8\textwidth}{
\includegraphics[width=1.8cm]{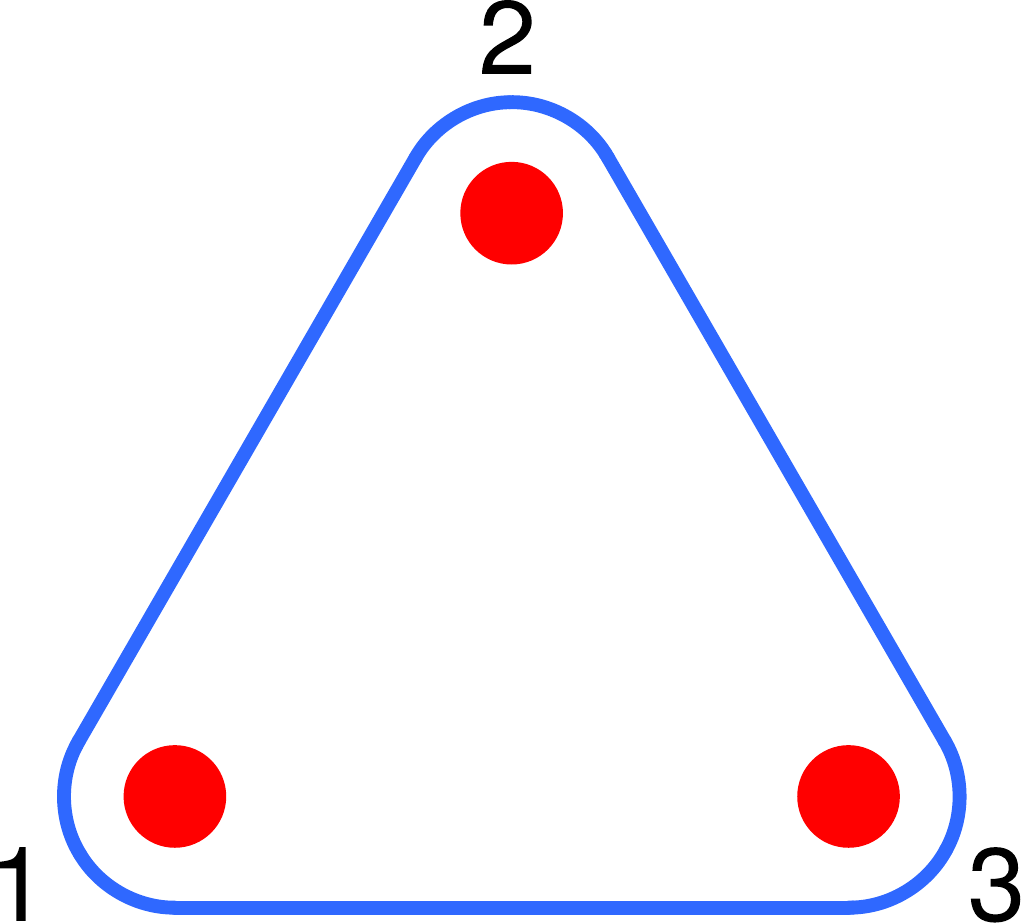} \quad\; \includegraphics[width=1.8cm]{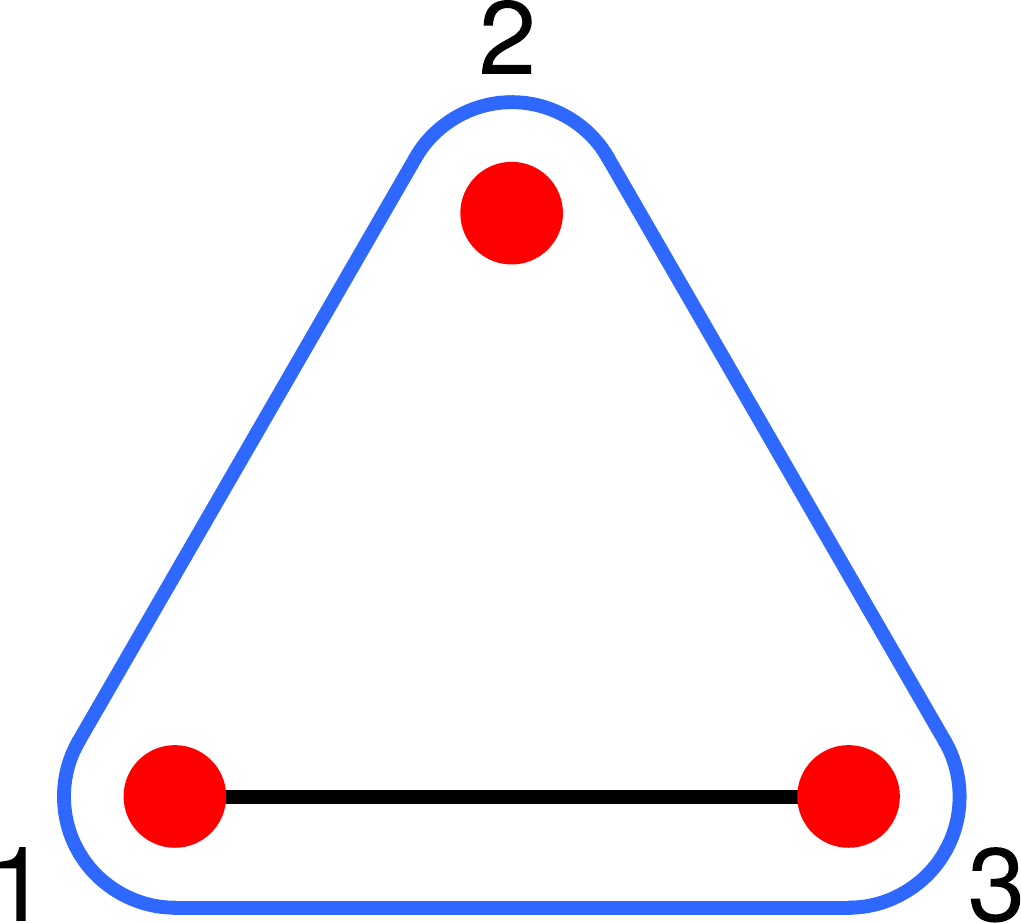} \quad\;
\includegraphics[width=1.8cm]{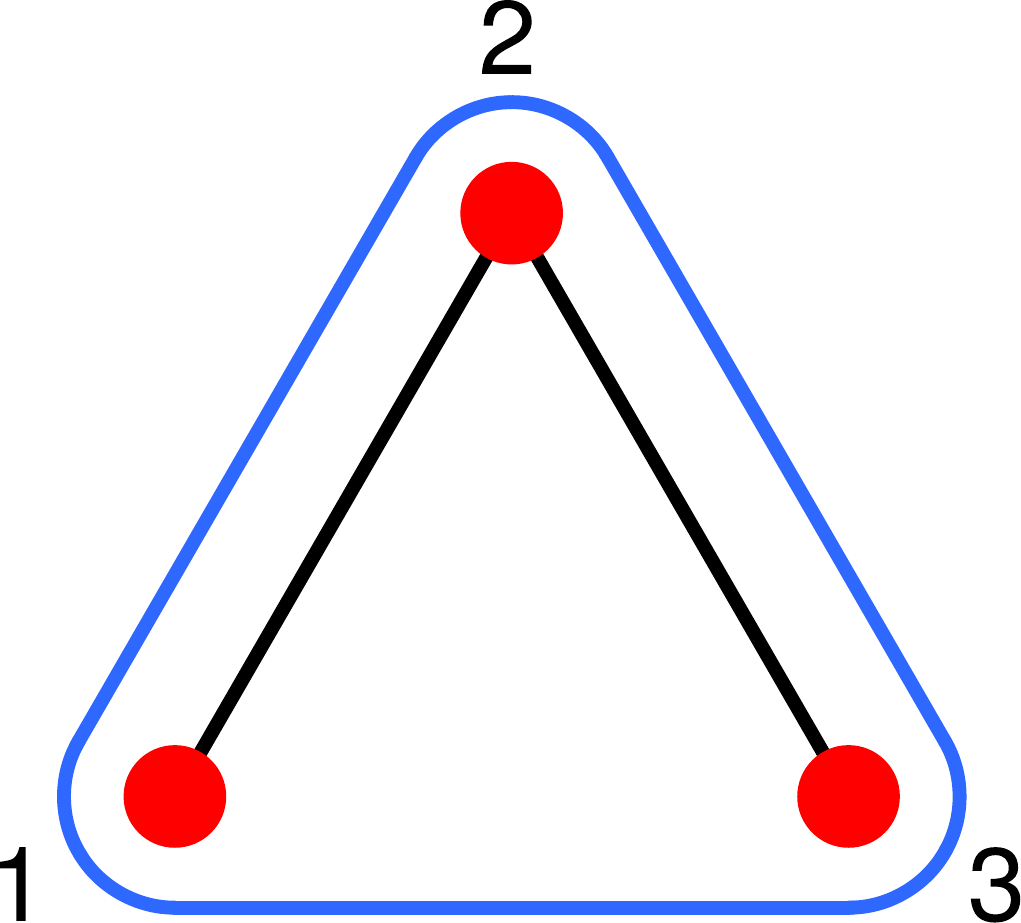} \quad\; \includegraphics[width=1.8cm]{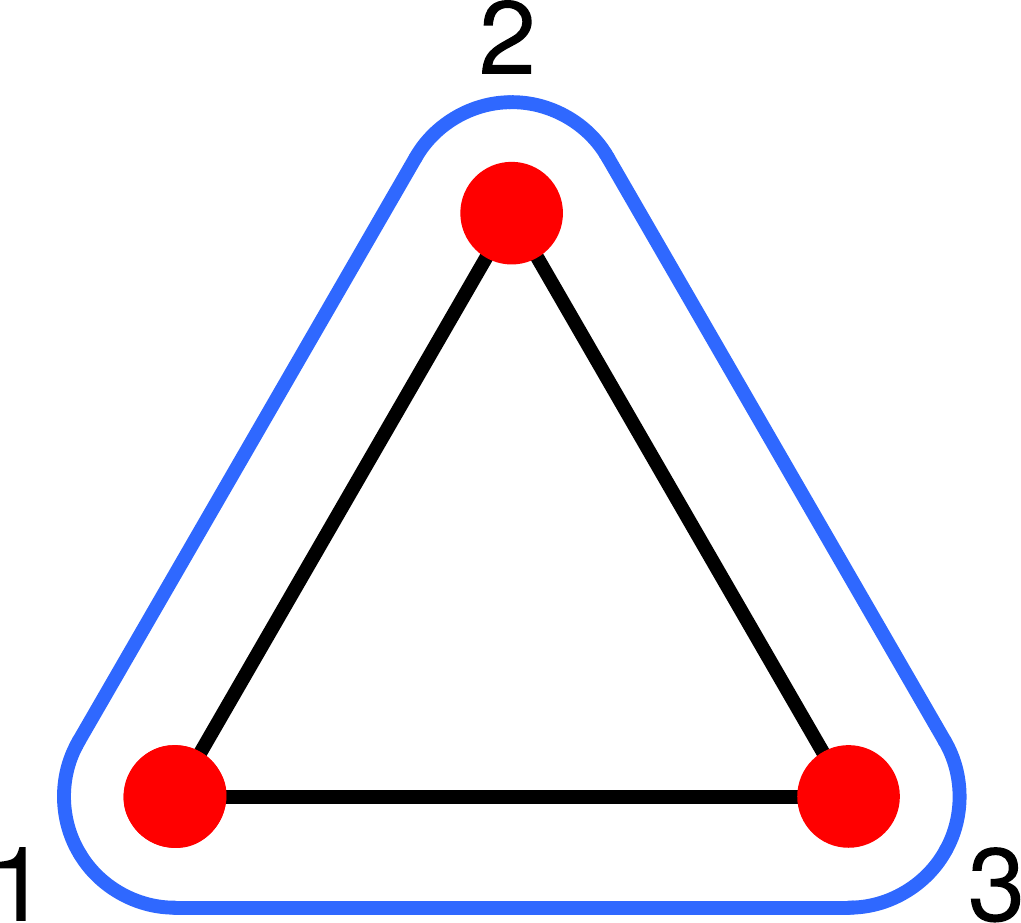} 
} & \parbox[c][1.5cm]{.07\textwidth}{$\sim$ 0.53 \\ $\sim$ 76.0\%} \\ \hline
\end{tabularx}
\caption{Table of one SLOCC class and two LU classes of 3-qutrit hypergraph multipartite entangled states.}
\label{tab:d3}
\end{table*}
\vspace{-11pt}
\end{center}
\begin{center}
\vspace{-11pt}
\begin{table*}[p]
\begin{tabularx}{\textwidth}{| c | c | c | c |}
\hline
\parbox[c][1.5cm]{.018\textwidth}{\rotatebox{90}{Class}} & \parbox[c]{.07\textwidth}{Schmidt ranks} & \parbox[c][1.5cm]{.8\textwidth}{Representatives} &  \parbox[c][1.5cm]{.077\textwidth}{Geom. measure/ w-noise tolerance } \\ \hline
{\bf 1} & \parbox[c]{.07\textwidth}{1$|$23\quad 4\\ 2$|$13\quad 4\\ 3$|$12\quad 4} & \parbox[c][2cm]{.8\textwidth}{
\includegraphics[width=1.8cm]{s12s23.pdf} \quad\; \includegraphics[width=1.8cm]{s12s23s13.pdf} \quad\;
\includegraphics[width=1.8cm]{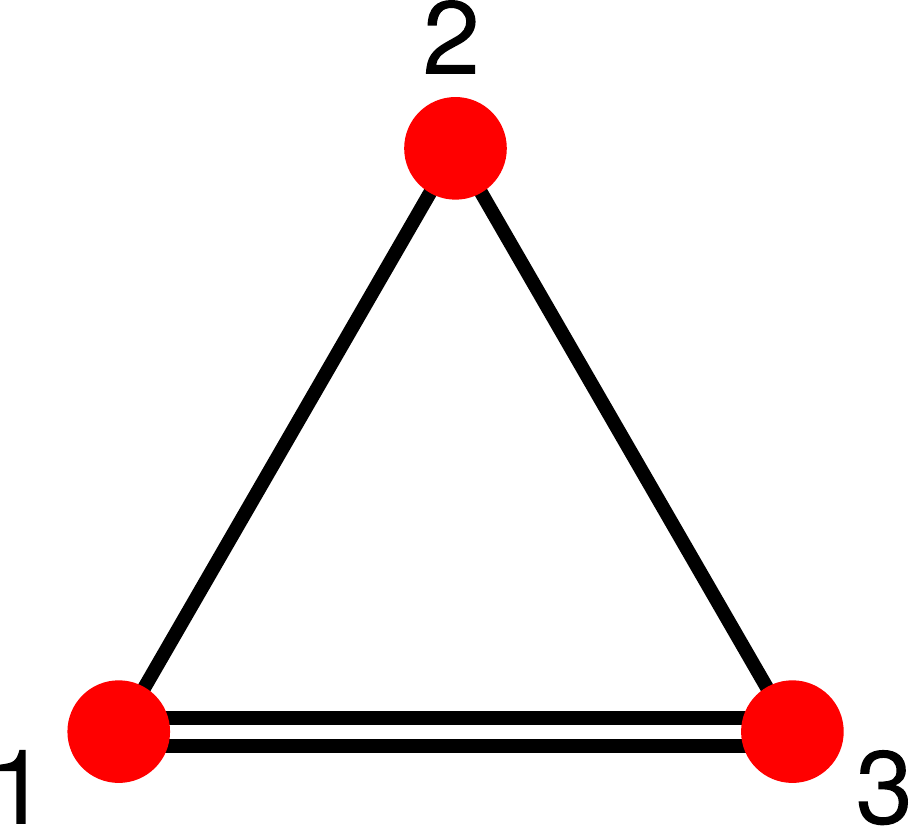} \quad\; \includegraphics[width=1.8cm]{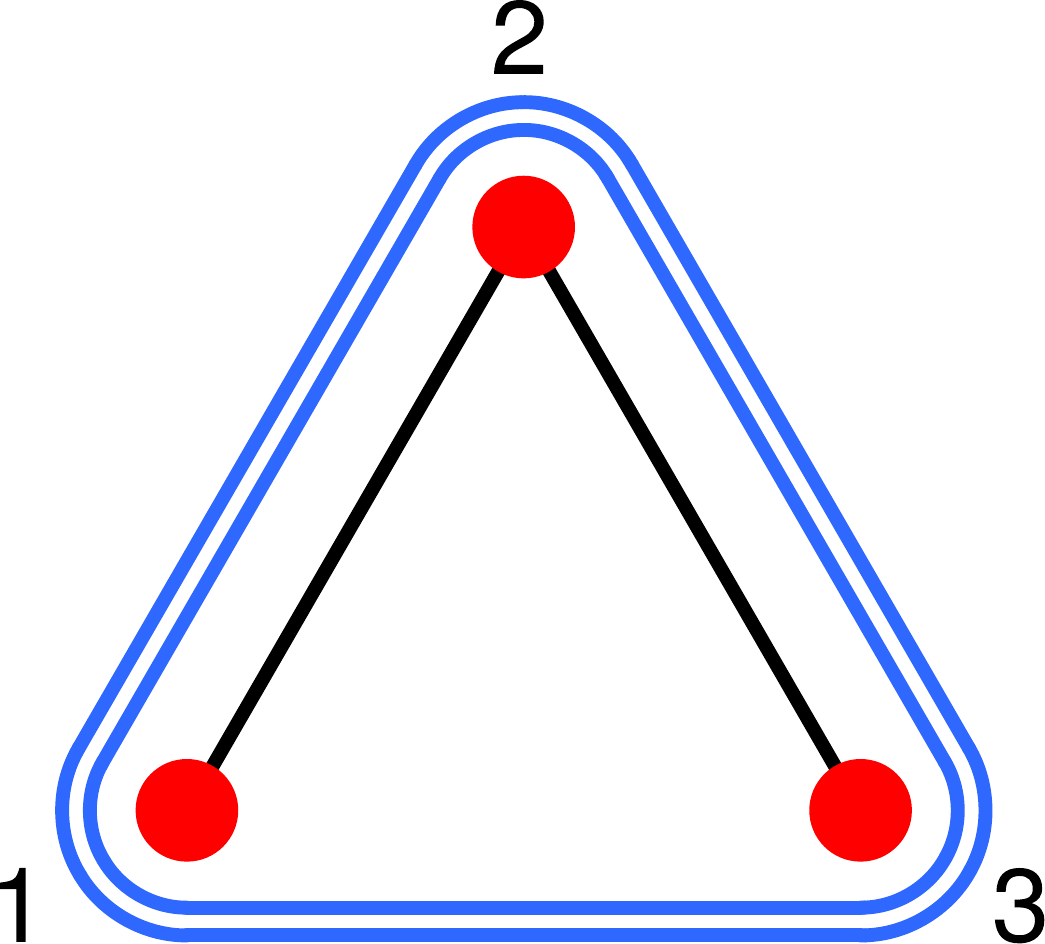} \quad\;
\includegraphics[width=1.8cm]{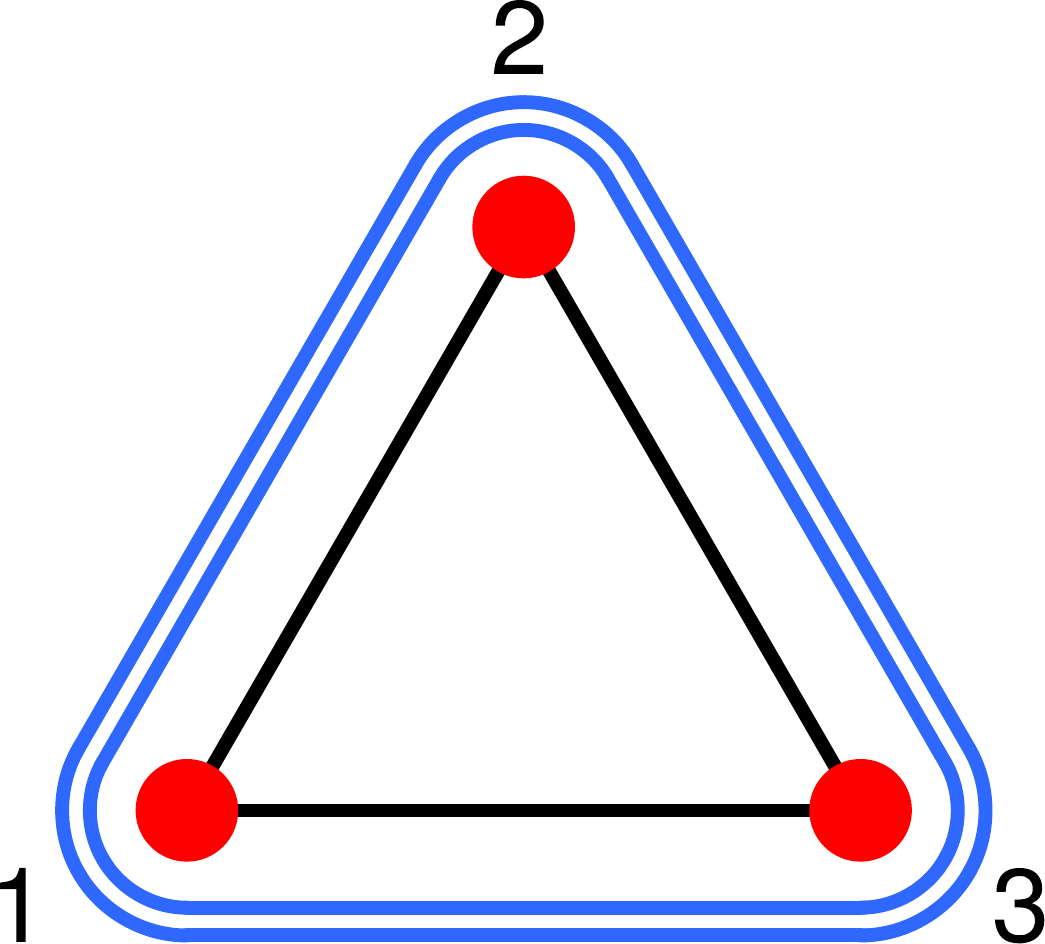} \quad\; \includegraphics[width=1.8cm]{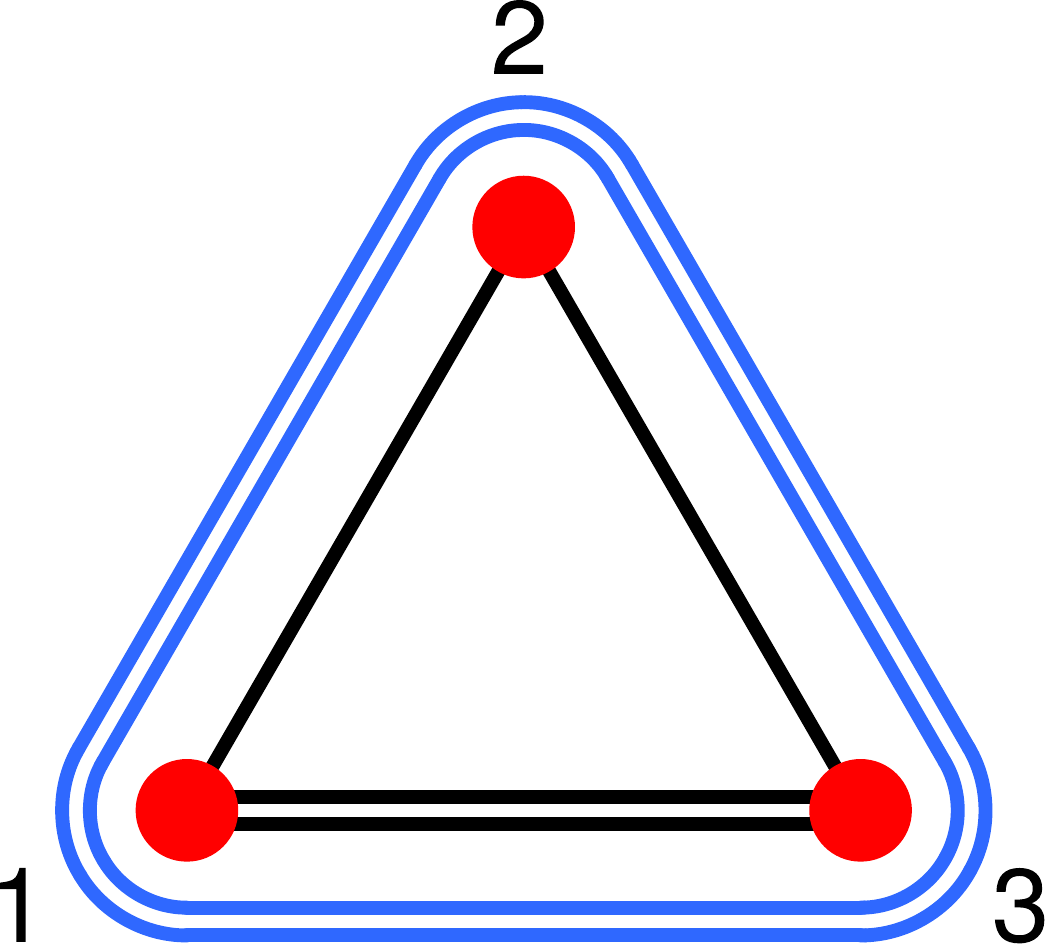}
} & \parbox[c][1.5cm]{.07\textwidth}{0.75 \\ $\sim$ 84.2\%} \\ \hline
{\bf 1'} & \parbox[c]{.07\textwidth}{1$|$23\quad 4\\ 2$|$13\quad 4\\ 3$|$12\quad 4} & \parbox[c][4.2cm]{.8\textwidth}{
\includegraphics[width=1.8cm]{s123.pdf} \quad\; \includegraphics[width=1.8cm]{s123s13.pdf}  \quad\;
\includegraphics[width=1.8cm]{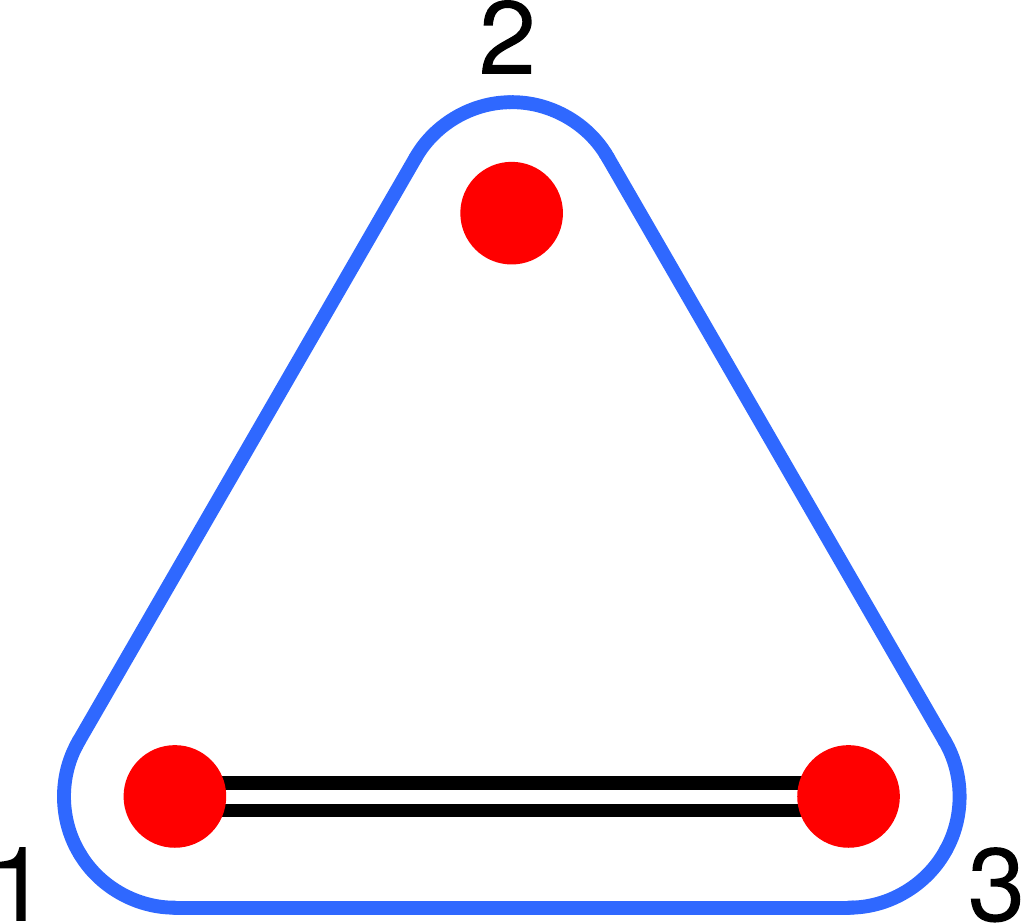} \quad\; \includegraphics[width=1.8cm]{s123s12s23.pdf} \quad\;
\includegraphics[width=1.8cm]{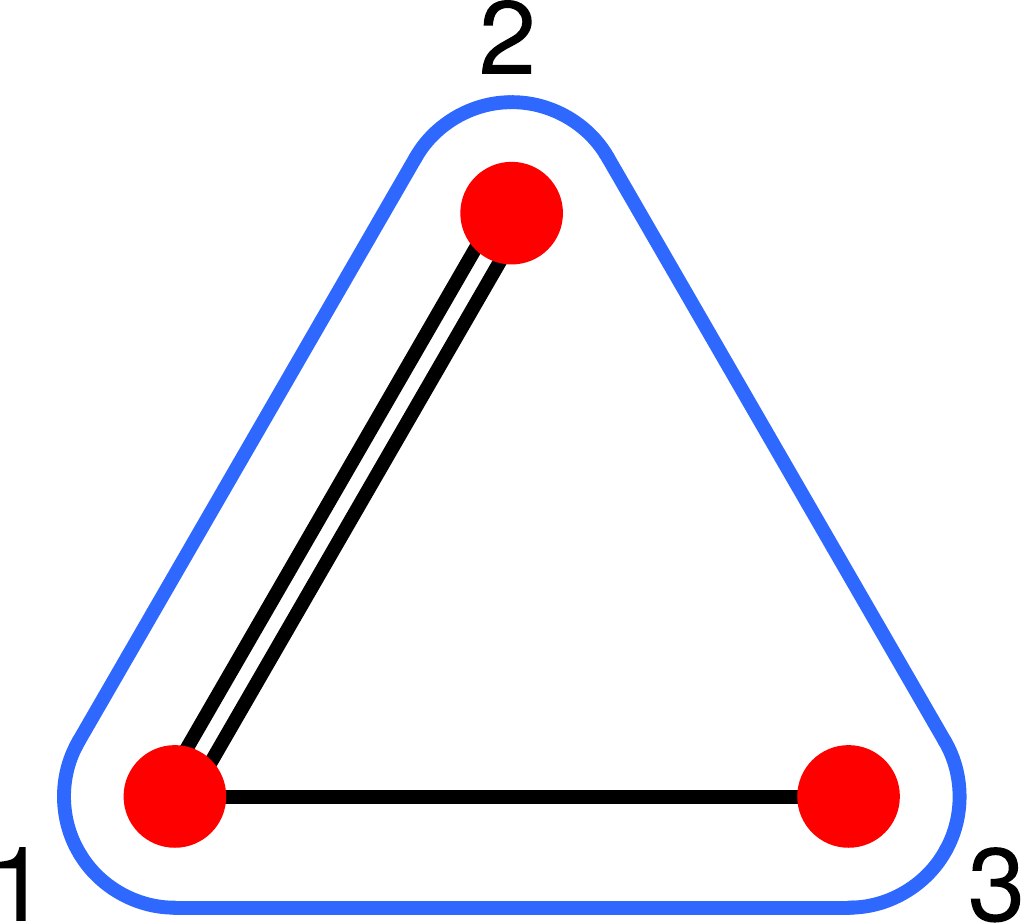} \quad\; \includegraphics[width=1.8cm]{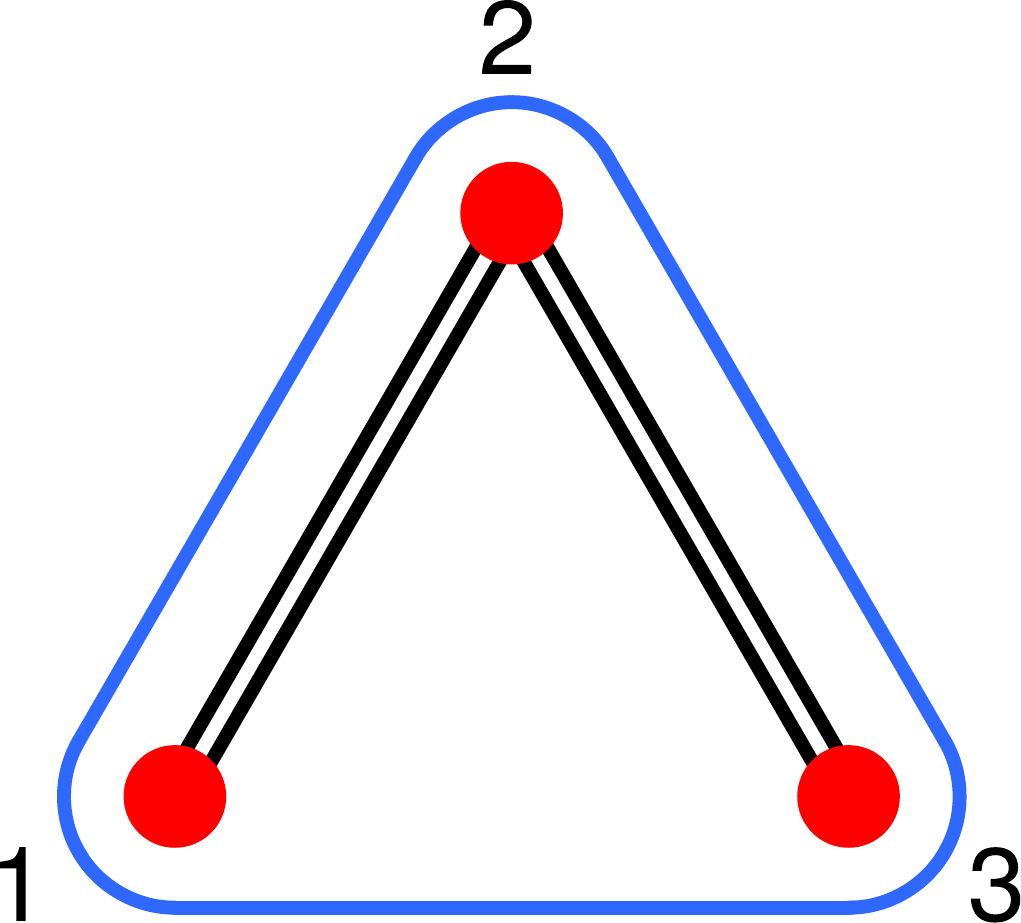} \\[.4cm]
\includegraphics[width=1.8cm]{s123s12s23s13.pdf} \quad\; \includegraphics[width=1.8cm]{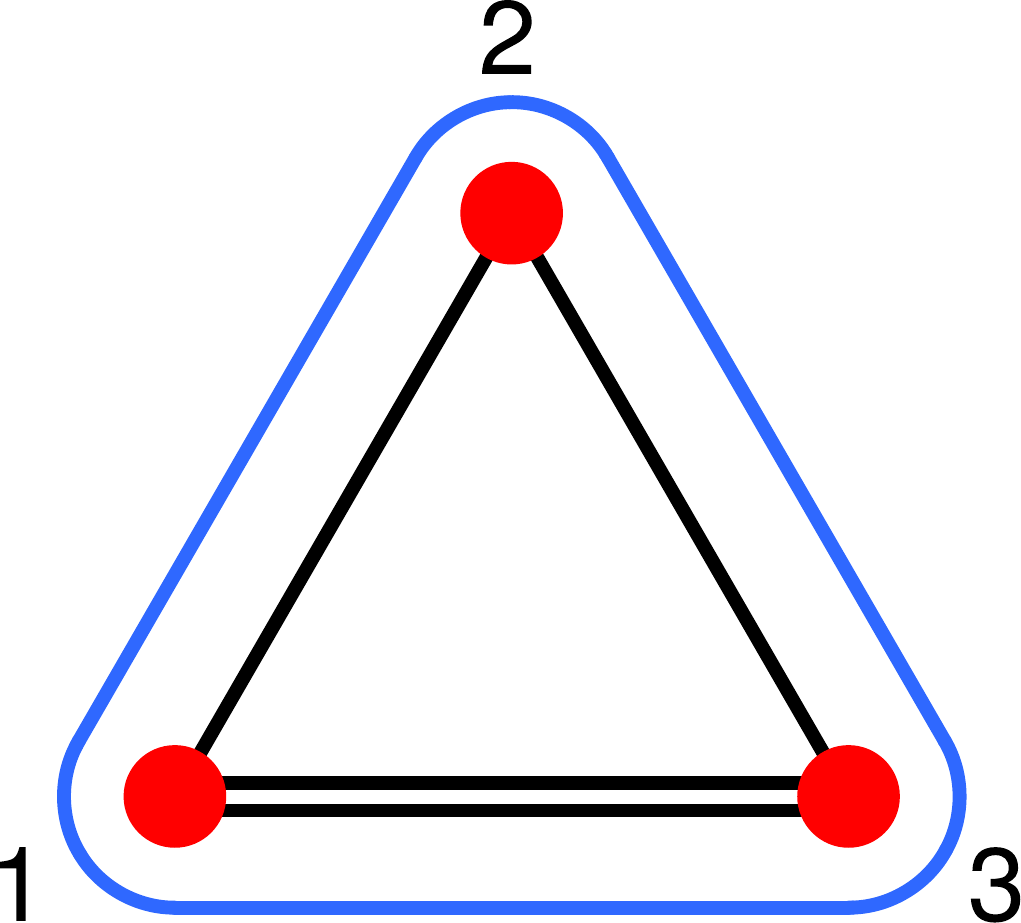} \quad\;
\includegraphics[width=1.8cm]{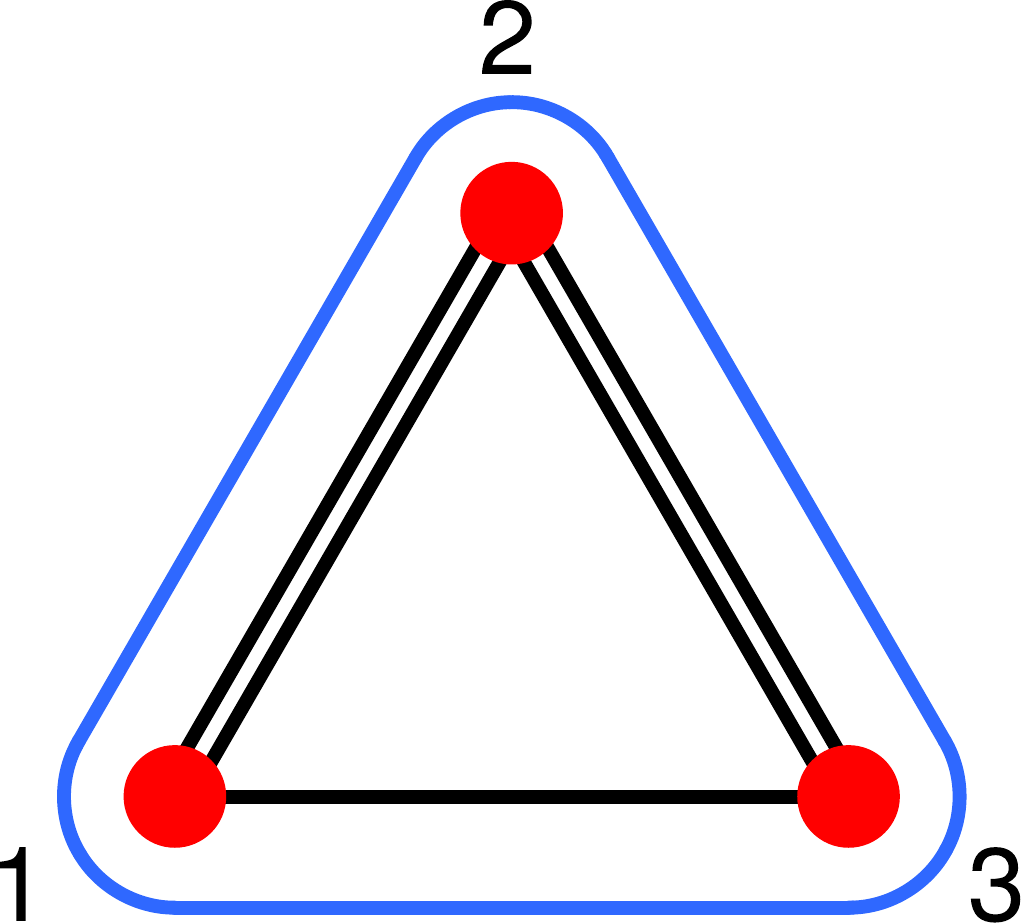} \quad\; \includegraphics[width=1.8cm]{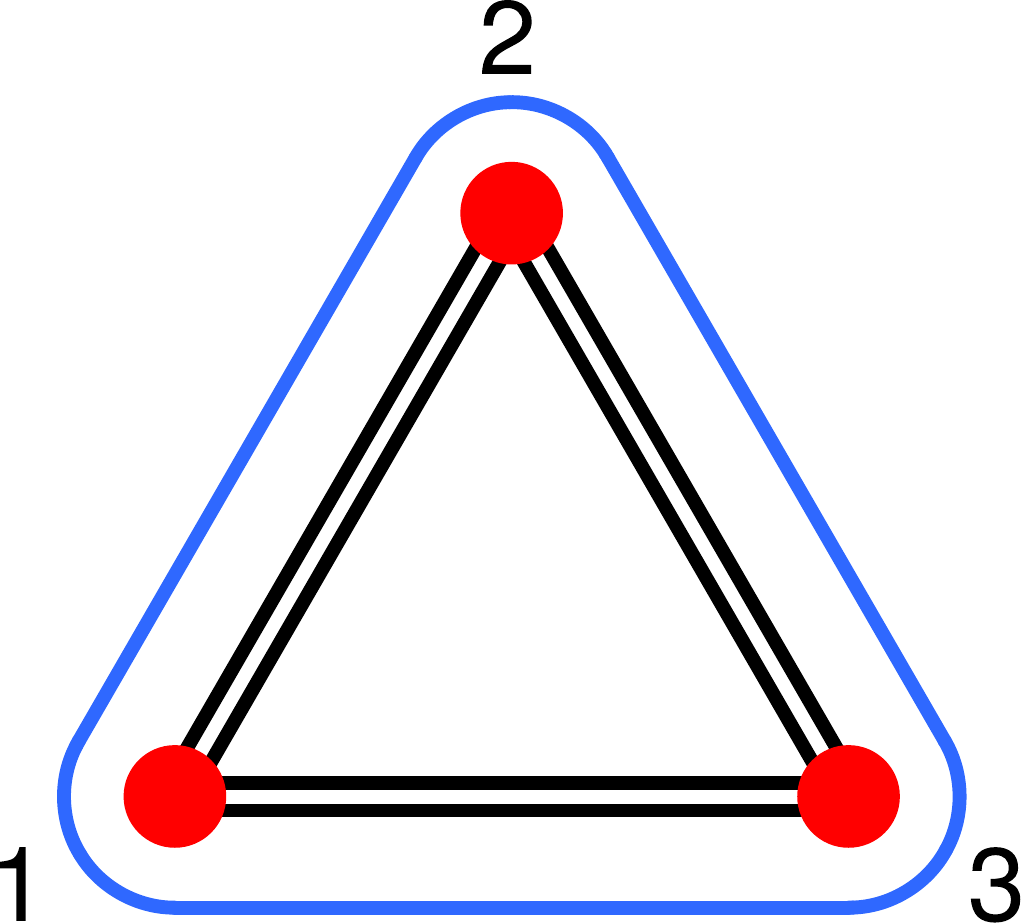}
} & \parbox[c][1.5cm]{.07\textwidth}{$\sim$ 0.58 \\ $\sim$ 87.1\%} \\ \hline

\multirow{ 2}{*}{\bf 2} & \parbox[c]{.07\textwidth}{\vspace{.5cm}1$|$23\quad 2\\ 2$|$13\quad 2\\ 3$|$12\quad 2} & \parbox[c][2cm]{.8\textwidth}{
\includegraphics[width=1.8cm]{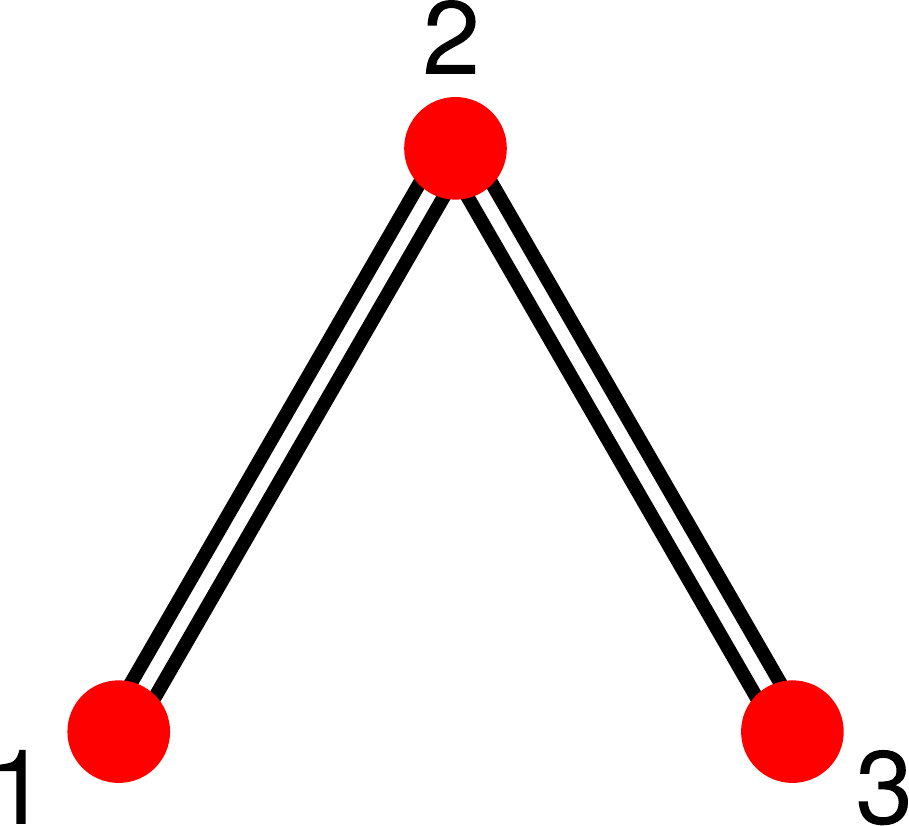} \quad\; \includegraphics[width=1.8cm]{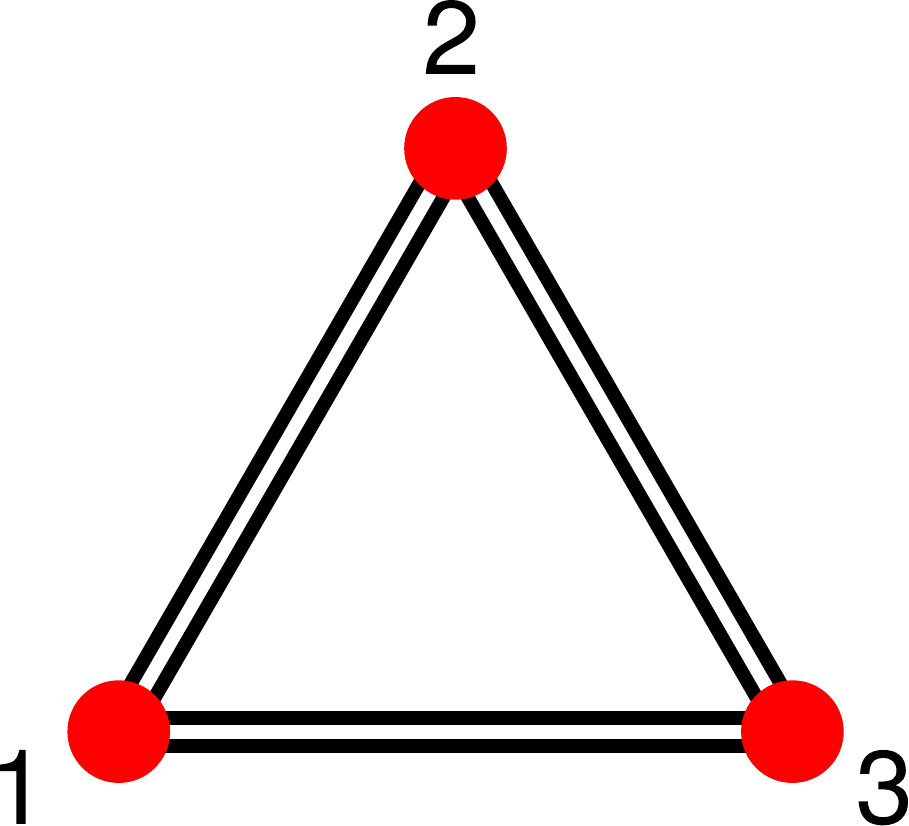}} 
& \parbox[c][1.5cm]{.07\textwidth}{0.50 \\ $\sim$ 91.4\%} \\ \cline{3-4}
&  & \parbox[c][2cm]{.8\textwidth}{
\includegraphics[width=1.8cm]{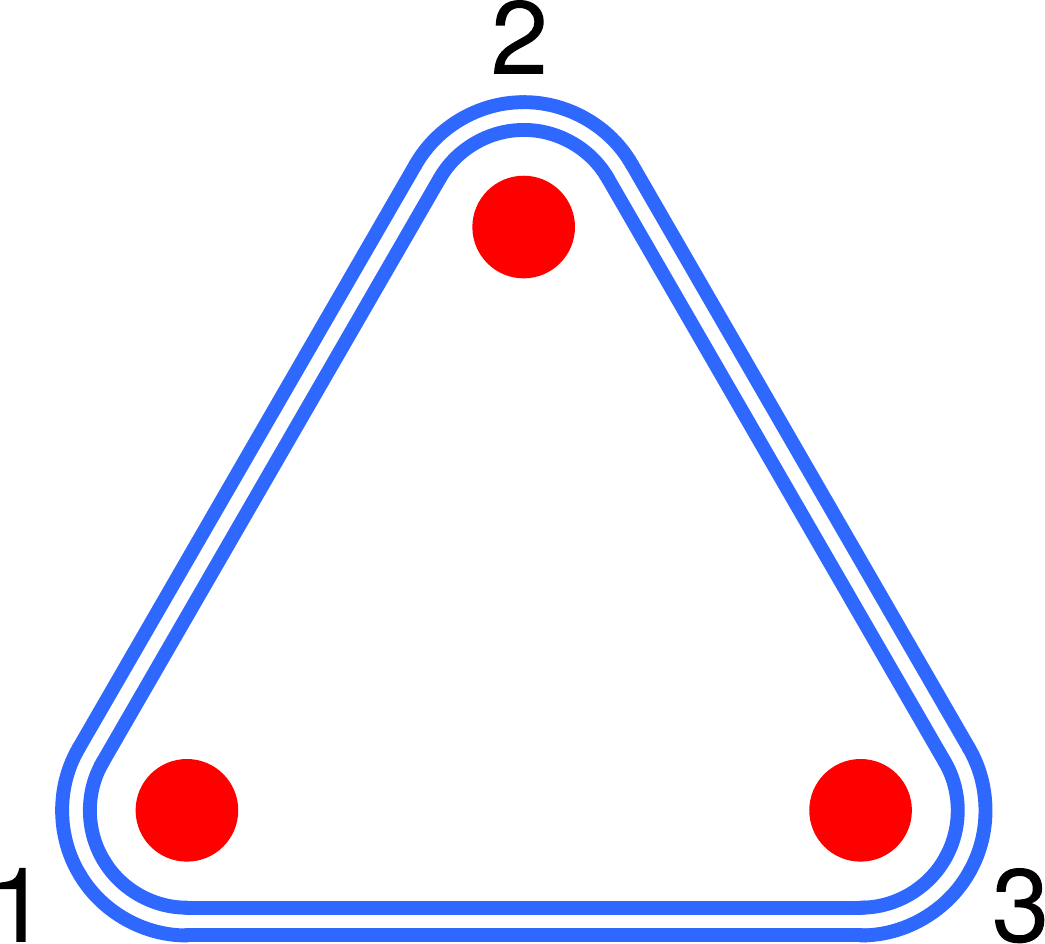} \quad\; \includegraphics[width=1.8cm]{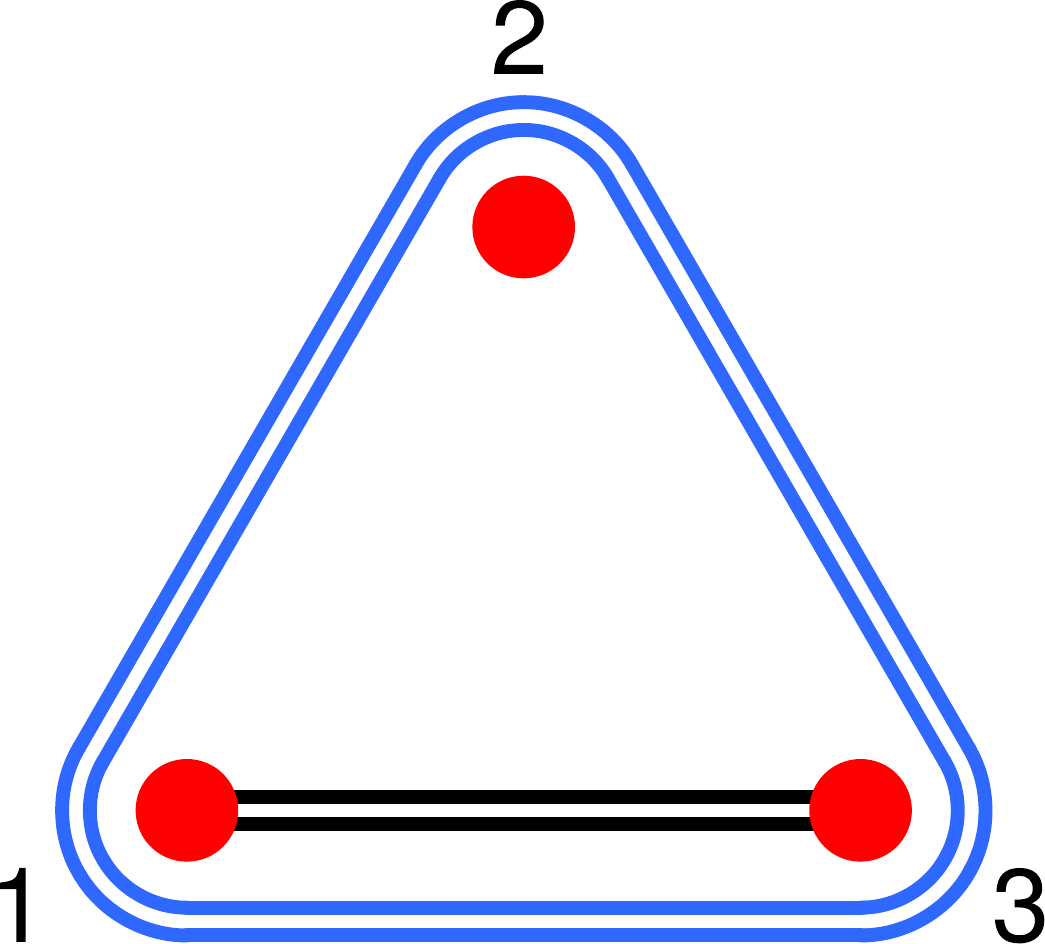} \quad\;
\includegraphics[width=1.8cm]{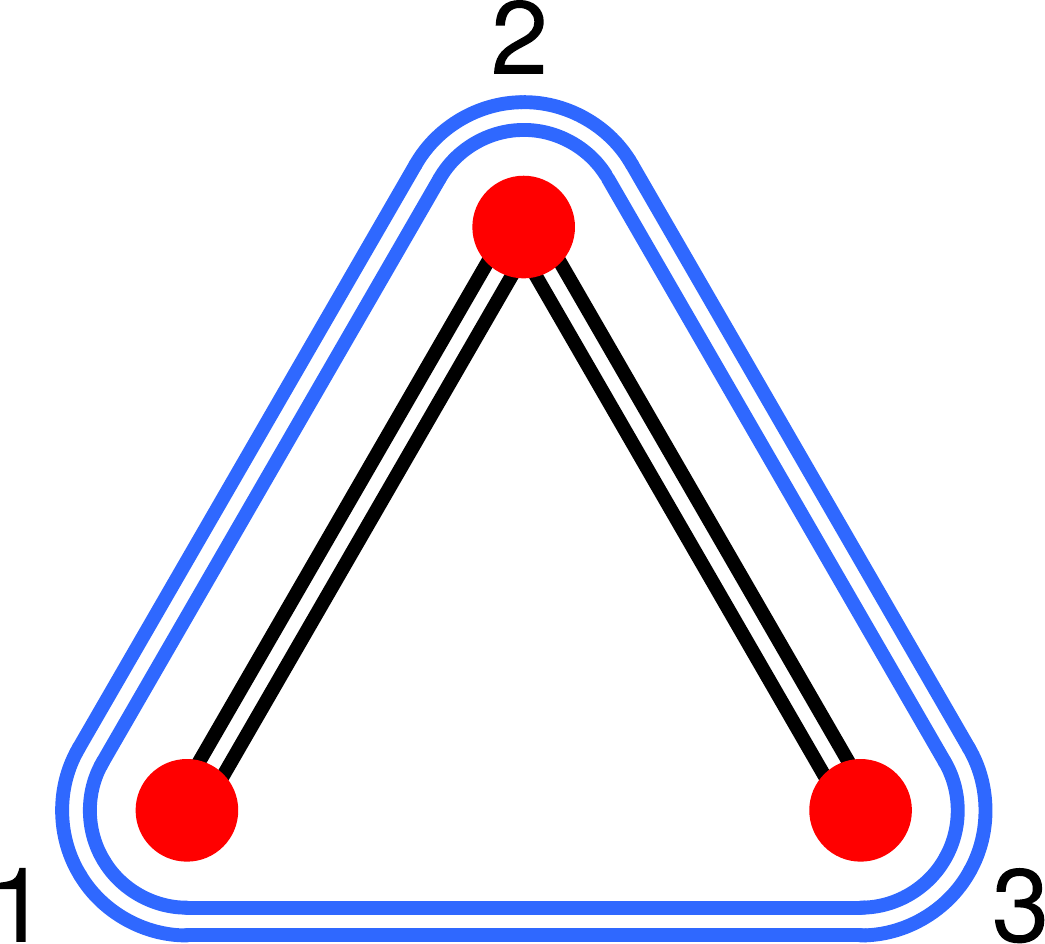} \quad\; \includegraphics[width=1.8cm]{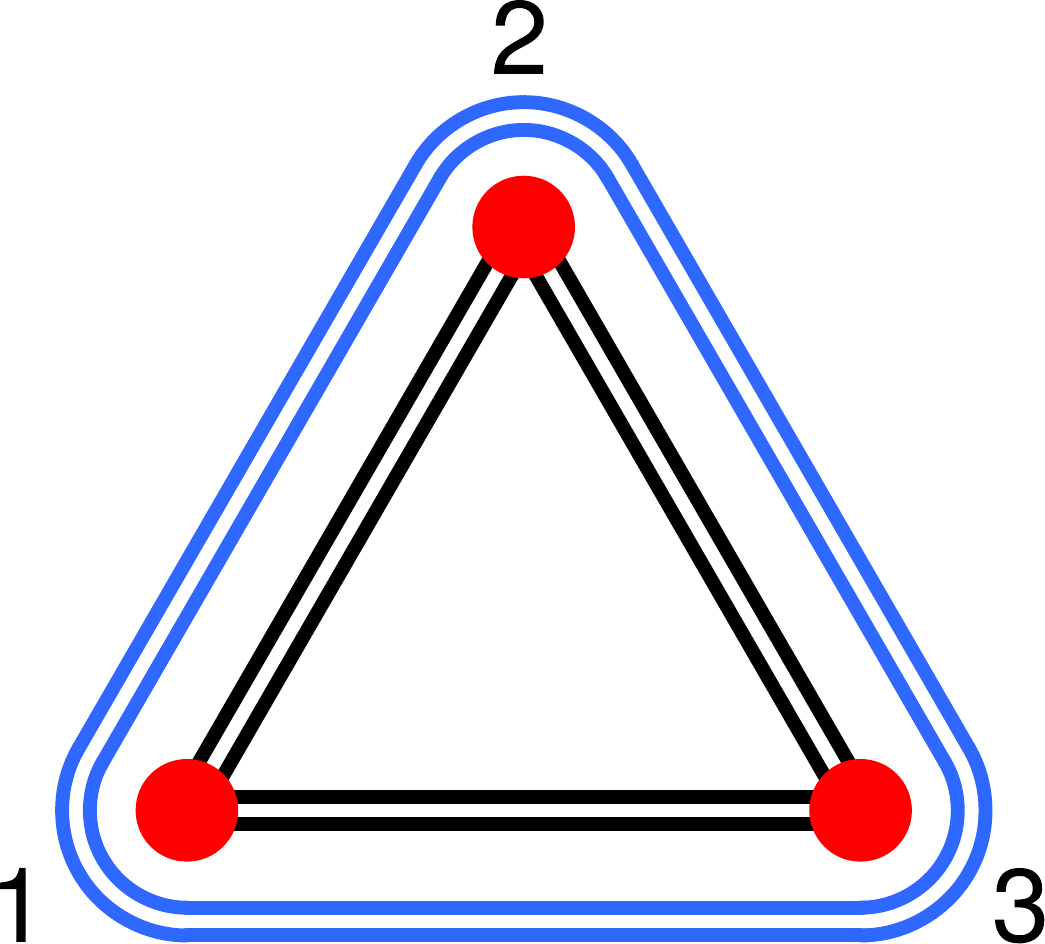}} & 
\parbox[c][1.5cm]{.07\textwidth}{$\sim$ 0.32 \\ $\sim$ 88.7\%} \\ \hline 
{\bf 3} & \parbox[c]{.07\textwidth}{1$|$23\quad 4\\ 2$|$13\quad 2\\ 3$|$12\quad 4} & \parbox[c][2cm]{.8\textwidth}{
\includegraphics[width=1.8cm]{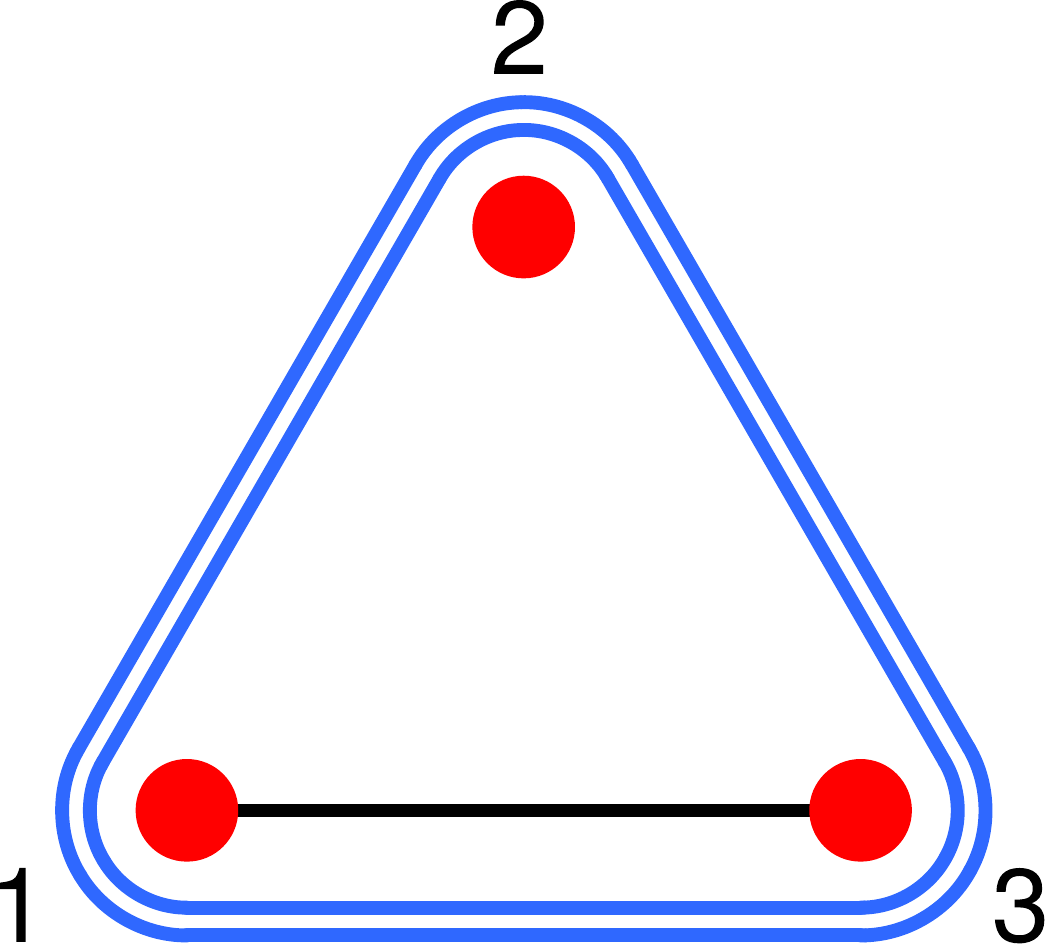} \quad\; \includegraphics[width=1.8cm]{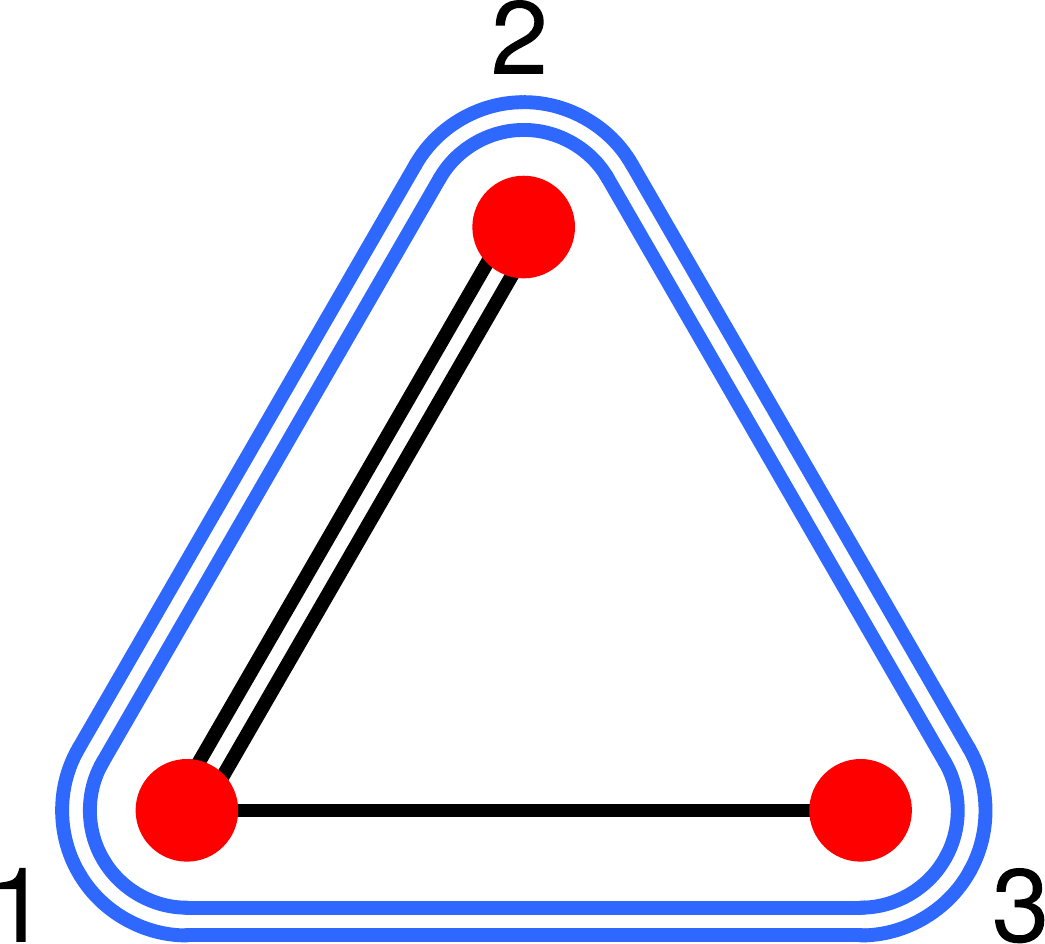} \quad\;
\includegraphics[width=1.8cm]{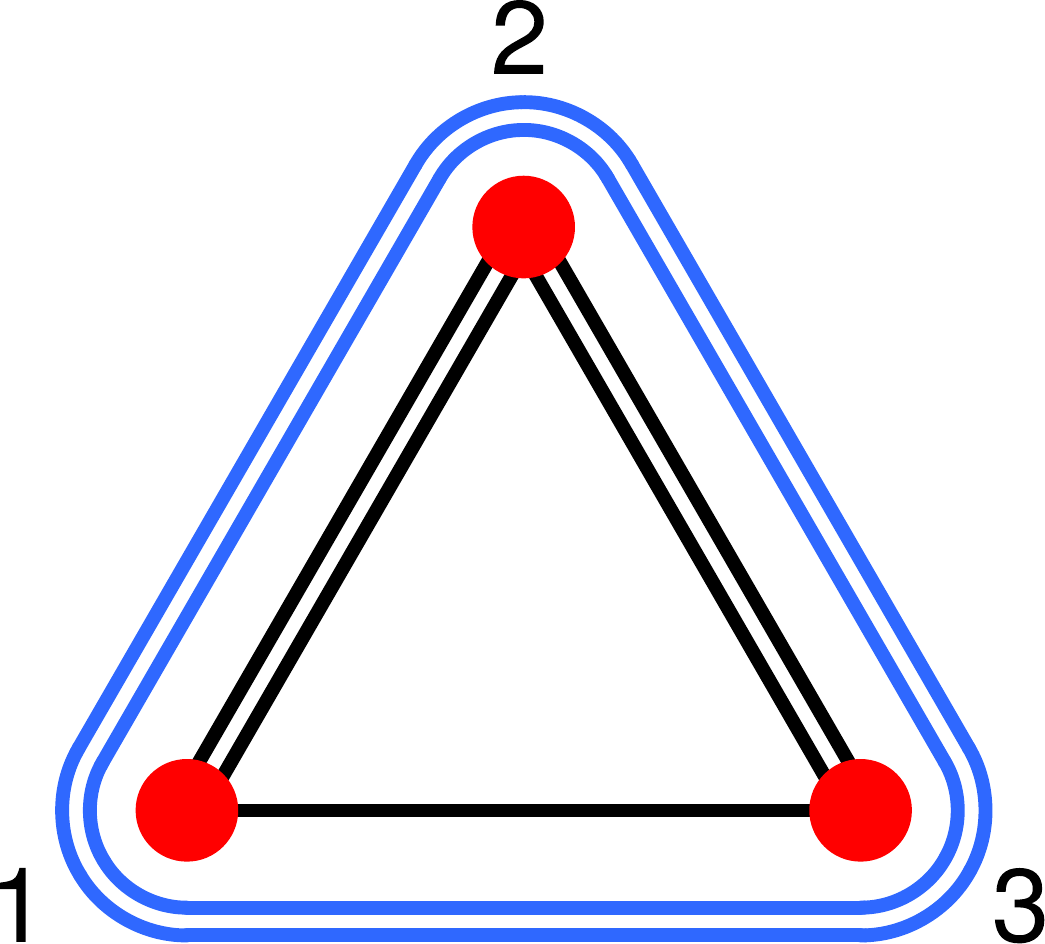}} & \parbox[c][1.5cm]{.07\textwidth}{0.75 \\ $\sim$ 86.1\%}  \\ \hline

{\bf 4} & \parbox[c]{.07\textwidth}{1$|$23\quad 4\\ 2$|$13\quad 2\\ 3$|$12\quad 4} & \parbox[c][2cm]{.8\textwidth}{
\includegraphics[width=1.8cm]{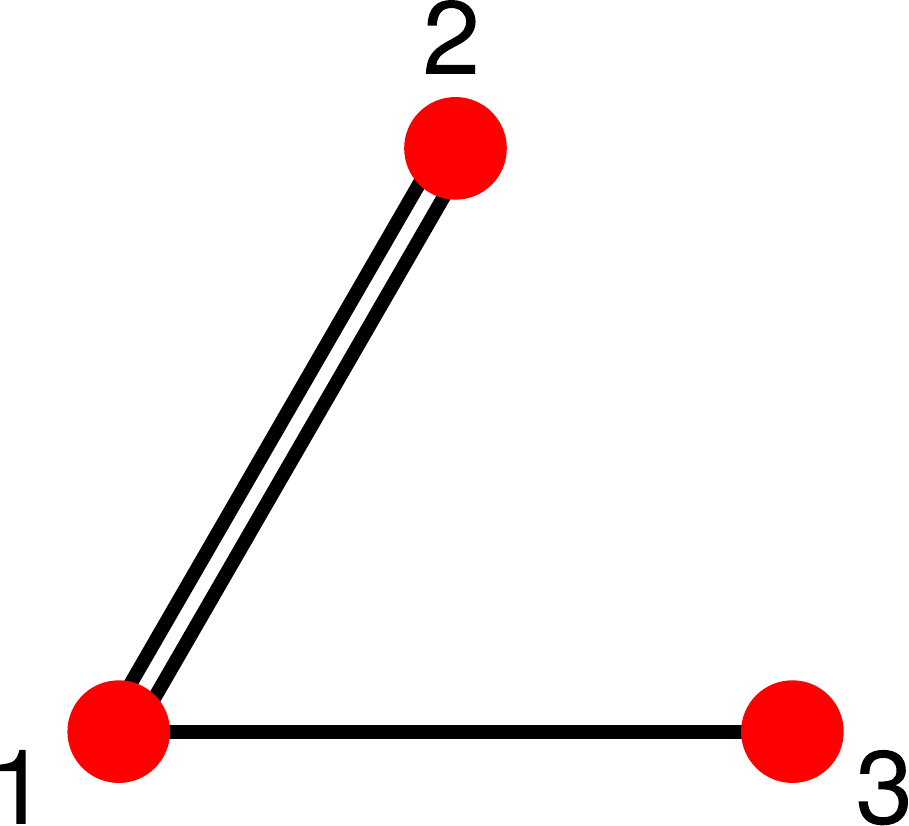} \quad\; \includegraphics[width=1.8cm]{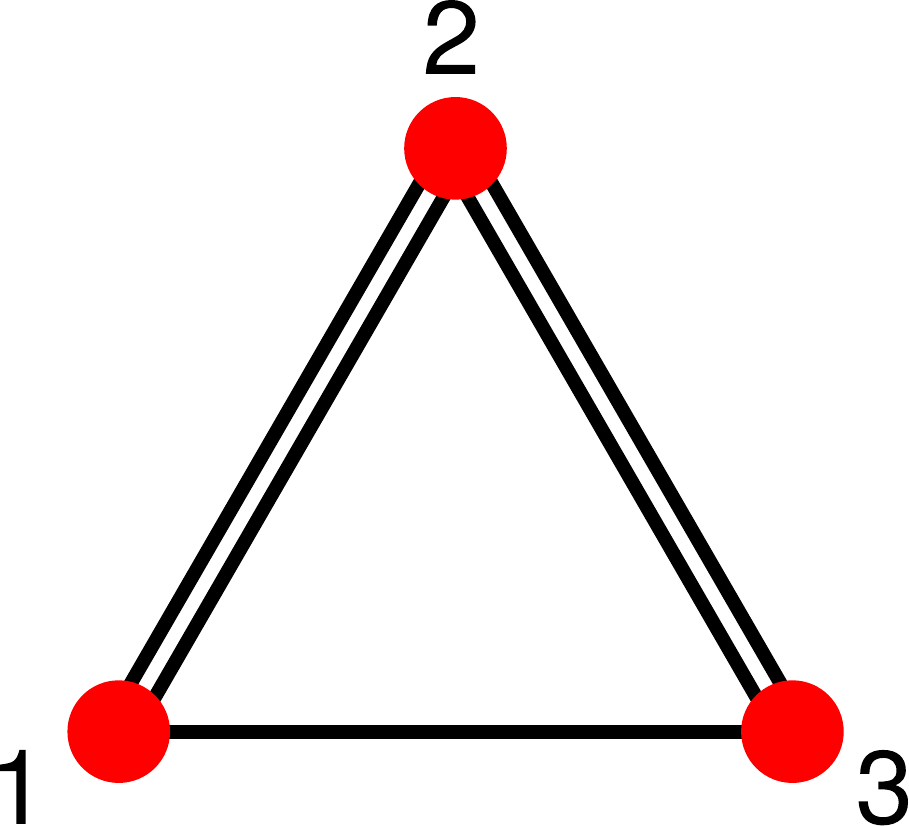}} & \parbox[c][1.5cm]{.07\textwidth}{0.75 \\ $\sim$ 88.8\%}  \\ \hline
\end{tabularx}
\caption{Table of SLOCC classes and LU classes of 3-ququart hypergraph multipartite entangled states.}
\label{tab:d4}
\end{table*}
\end{center} 
}
We now consider the special cases of a tripartite system with prime 
dimension $3$ and a tripartite system with the smallest non-prime 
dimension $4$ as examples. 

\subsection{Classification of $3\otimes 3\otimes 3$}

In the case of a tripartite system of qutrits, there is only one SLOCC equivalence
class of hypergraph states and two  LU equivalence classes: the GHZ state and the 
$3$-elementary hypergraph state. These two states are inequivalent by LU for 
presenting different values of geometric measure of entanglement and white noise tolerance in Table (\ref{tab:d3}). 

This classes can be derived as follows: Let us first consider the GHZ state.
{From} Appendix B it follows that the GHZ state can be converted to the graph 
state represented by the local complementation of the GHZ graph via local 
symplectic unitaries. {From} Proposition \ref{propclifford} we see that a
hyperedge of arbitrary multiplicity can be converted to an hyperedge of 
any other multiplicity via local symplectic permutations. Thus, the 
tripartite GHZ state is equivalent to any other tripartite graph 
state via local symplectic unitaries. 

If we consider the elementary hypergraph state, a $3$-hyperdege can be 
converted to another $3$-hyperedge of arbitrary multiplicity via symplectic 
permutations. In addition, the $3$-elementary hypergraph is equivalent to any other 
$3$-hypergraph, since edges ($2$-hyperedges) of arbitrary multiplicities 
can be created via repeated application of the $X^{\dagger}$ gate in a 
neighbouring qutrit. 

Finally, in order to show the SLOCC equivalence, local invertible operations 
connecting these two LU subclasses can be achieved by applying $A_1$ to one 
of the qutrits of the graph state and $A_{2,3}$ to the other two, where
\begin{eqnarray}
A_{1} = \frac{1}{4\sqrt{3}}\left(\begin{array}{c c c}
			{-2\sqrt{3}-2i}&{4i}&{4i} \\
			{-4\sqrt{3}+4i}&{\sqrt{3}+i}&{-5\sqrt{3}+7i} \\
			{-6\sqrt{3}-2i}&{-\sqrt{3}-5i}&{-\sqrt{3}+7i}
			\end{array}\right)
\end{eqnarray}
and
\begin{eqnarray}
A_{2,3} = \frac{1}{3}\left(\begin{array}{c c c}
			{e^{i2\pi/3}}&{1}&{1} \\
			{\sqrt{3}e^{i\pi/6}}&{\sqrt{3}e^{i\pi/6}}&{\sqrt{3}e^{i\pi/6}} \\
			{e^{i2\pi/3}}&{e^{i2\pi/3}}&{\frac{5-\sqrt{3}i}{2}}
			\end{array}\right).
\end{eqnarray}
This local operations were found with the help of the tool\# 3 (numeric optimization program described in the previous section), which gives in this case full product basis for right subspaces of all states from Table~\ref{tab:d3}.  

\subsection{Classification of $4\otimes 4\otimes 4$}

In the case of a tripartite system of ququarts, there are five SLOCC and 
six LU equivalence classes of hypergraph states. All possible states with respect to permutations and equivalence of edge multiplicities (local Clifford permutation $S$ 
converts the multiplicity of the $3$-hyperedge from $1$ to $3$ (since $3=3^{-1}$ modulo $4$, 
see Proposition \ref{propclifford}), see also Fig.~\ref{fig:Uni}b.), are shown in the Table~\ref{tab:d4} and the interconversion between representatives within the same class are explained in detail in what follows.

\subsubsection{Class $1$}

Class $1$ contains hypergraph states with at least two edges of multiplicity $1$ and with either no hyperedges, or with hyperedge of multiplicity $2$. All these states belong to the same LU-equivalence class.

LU-equivalence among first three state of class $1$ (see Table~\ref{tab:d4}) is governed by  
standard local complementation operations, which can be used to create a new edge of multiplicity $1$ in the neighborhood of qudit $2$, while applying these operations once more generates an edge of multiplicity $2$ in the neighborhood of qudit $2$ (see Appendix B for more details). The same local complementation is responsible for LU equivalence among three last states the first LU class.

To prove LU equivalence between these two subgroups of states (with no hyperedge and with $2$-hyperedge) we find with the Tool\# 3 explicit form of their MEBs, which appear to consist of product vectors. It can be shown then that local transformation between these states is unitary. Here we present such local unitary for transformation from the Figure~\ref{fig:Uni}e 
\begin{eqnarray}
U_{1,2,3} = \frac{1}{2}\left(\begin{array}{c c c c}
			{1+i}&{0}&{1-i}&{0} \\
			{0}&{0}&{0}&{-2} \\
			{1-i}&{0}&{1+i}&{0} \\
			{0}&{-2}&{0}&{0} 
			\end{array}\right).
\end{eqnarray}

\subsubsection{Class $1'$}
Class $1'$ contains all hypergraph states which have $3$-hyperedge of multiplicity $1$. 
LU equivalence of the states within this class is governed by the unitary $(X^{\dagger})^m$, 
which, when allied to some qudit, generates edges of multiplicity $m$ on the neighbourhood 
of the qudit (see Lemma \ref{lemmacommutation}). 

\subsubsection{Class $2$}

Class $2$ consists of two LU equivalence classes. 
The representative of the first LU-equivalence class are the graph 
states composed of two and three edges of multiplicity $2$, though the representatives of the second LU class are the hypergraph state with a $3$-hyperedge of multiplicity $2$ with possible edges of multiplicity $2$.

We can perform some form of ``local complementation"\, between two states from the first LU class by applying the following unitaries in the basis $\{|p_0\rangle,|p_1\rangle,|p_2\rangle,|p_3\rangle\}$:
\begin{eqnarray*}
U_{1,3} = \frac{1}{\sqrt{2}}\left(\begin{array}{c c c c}
			{1}&{0}&{i}&{0} \\
			{0}&{\sqrt{2}}&{0}&{0} \\
			{-i}&{0}&{-1}&{0} \\
			{0}&{0}&{0}&{\sqrt{2}} 
			\end{array}\right); \ \ U_2 = \left(\begin{array}{c c c c}
			{i}&{0}&{0}&{0} \\
			{0}&{1}&{0}&{0} \\
			{0}&{0}&{1}&{0} \\
			{0}&{0}&{0}&{1} 
			\end{array}\right).
\end{eqnarray*}

Applying the local $(X^{\dagger})^m$ unitary to some qudit of the states from the second LU class generates edges of multiplicity $2m$ (i.e., $0$ or $2$) on the neighbourhood of that qudit. 

Using the Tool \#3 one can find the local operation corresponding to SLOCC equivalence between these LU classes. For the representatives shown on the Fig.~\ref{fig:Uni}d the corresponding LO is
\begin{eqnarray}
A_{1,2,3} = \frac{1}{2}\left(\begin{array}{c c c c}
			{-i(1+\sqrt[3]{4})}&{0}&{(1-\sqrt[3]{4})}&{0} \\
			{0}&{2}&{0}&{0} \\
			{i}&{0}&{-1}&{0} \\
			{0}&{0}&{0}&{2} 
			\end{array}\right).
\end{eqnarray}
One can easily check that $A_{1,2,3}$ is invertible, but not unitary. To show that there is no local unitary transformation possible, one can look at the entanglement measures for these LU classes (see Table~\ref{tab:d4}).

\begin{figure}[t]
\begin{center}
\begin{tabular}{c c c c}
\parbox[c]{0.1cm}{\vspace{-3cm}a)} & \includegraphics[width=1.8cm]{s123.pdf} & \parbox[c]{1.1cm}{\vspace{-2cm}\Large$\xrightarrow{(X^{\dagger}_2)^2}$} & \includegraphics[width=1.8cm]{s123d13.pdf}\\[.4cm]
\parbox[c]{0.1cm}{\vspace{-3cm}b)} & \includegraphics[width=1.8cm]{s123.pdf} & \parbox[c]{1.1cm}{\vspace{-2cm}\Large$\xrightarrow{\;\;\, S_3 \;\;\,}$} & \includegraphics[width=1.8cm]{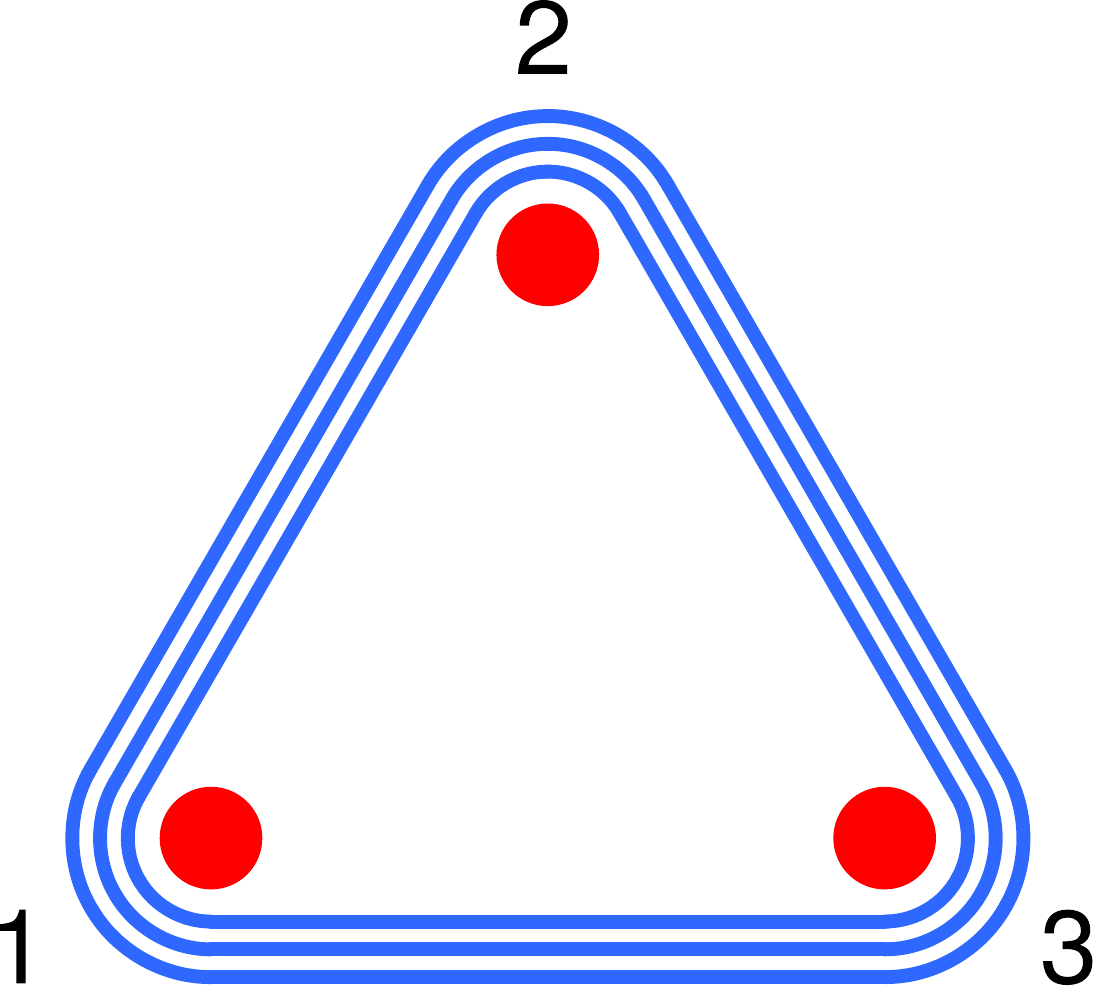}\\[.4cm]
\parbox[c]{0.1cm}{\vspace{-3cm}c)} & \includegraphics[width=1.8cm]{d123.pdf} & \parbox[c]{1.1cm}{\vspace{-2cm}\Large$\xrightarrow{\;\; X^{\dagger}_2 \,\;\;}$} & \includegraphics[width=1.8cm]{d123d13.pdf}\\[.4cm]
\parbox[c]{0.1cm}{\vspace{-3cm}d)} & \includegraphics[width=1.8cm]{d12d23d13.pdf} & \parbox[c]{1.1cm}{\vspace{-2cm}\Large$\xrightarrow{\,A^{\otimes 3}_{1,2,3}\,}$} & \includegraphics[width=1.8cm]{d123.pdf}\\[.4cm]
\parbox[c]{0.1cm}{\vspace{-3cm}e)} & \includegraphics[width=1.8cm]{s12s23s13.pdf} & \parbox[c]{1.1cm}{\vspace{-2cm}\Large$\xrightarrow{\,U^{\otimes 3}_{1,2,3}\,}$} & \includegraphics[width=1.8cm]{d123s12s23s13.pdf}
\end{tabular}
\vspace{-16pt}
\end{center}
\caption{SLOCC equivalence among representatives of the same class.}
\label{fig:Uni}
\end{figure}

\subsubsection{Class $3$}

The representatives of class $3$ are the elementary hypergraph 
states with a $3$-hyperedge of multiplicity $2$, one edge of 
multiplicity $1$ and possible edges of multiplicity $2$. These three states are in the same LU class
and the local transformation between them is $(X^{\dagger})$ applied on one of the qudits.

\subsubsection{Class $4$}

The representatives of class $4$ are graph states composed 
of one or two edges of multiplicity $2$ and one edge of multiplicity $1$. 
Applying the local unitaries $U_1=(S(1,1,0))^4$, $U_2= S(1,0,1)$, $U_3=S(1,1,0)$ to the first state creates an edge of multiplicity $2$ between qudits $2$ and $3$.

\subsubsection{SLOCC-inequivalence of Classes 1-4,1'}

To prove the SLOCC-inequivalence of states of most of the classes it is sufficient to look at their Schmidt ranks for each bipartition (see Table~\ref{tab:d4}). Exceptions are pairs of classes $1,1'$ and $3,4$.
 
To prove that 
there is no SLOCC transformation between states from classes $3$ and $4$ let us consider the vectors from the right subspace for bipartition $2|13$ for two representatives from each class.
From the Schmidt decomposition of the state from class $3$ one finds directly 
that there is at least one product vectors in the right subspace of parties $13$, 
i.e. MEB comtains at least one product vector. For the state from class $4$ we can prove that in the corresponding subspace there are no product vectors in the MEB using the Tool \#1.
Thus from Lemma~\ref{l:slocc-ineq} it follows that these states belong to different
SLOCC classes. 

Unfortunately, we were not able to prove SLOCC-inequivalence of states from classes $1$ and $1'$ using the tools presented above. In fact, using the Tool \#3 we found that the states from class $1$ have a full product basis in their right subspace for each bipartition and the Tool \#2 showed that for the states from class $1'$ there are states with PPT and full rank in their right subspace. However, the optimal value $\epsilon$ of the SDP of the Tool \#2 for the states in class $1'$ was in order of $10^{-5}$. Besides, the direct numerical search for SLOCC transformation bringing a states in class $1$ to some state in class $1'$ returned states of fidelity of almost $1$, though the numerical search for SLOCC transformation in the opposite direction, from a state in $1'$ to some state in $1$, succeed in returning states of fidelity of only $0.875$. This difference in fidelities of local transformations in different directions is typical for the three-qubit states of GHZ and W classes, which suggests that classes $1$ and $1'$ are inequivalent.           

\section{Conclusions} 

In this work we generalized the class of hypergraph states to systems 
of arbitrary finite dimensions. For the special class of elementary
hypergraph states we obtained the full SLOCC classification in terms 
of the greatest common divisor, which also governs other properties 
such as the ranks of reduced states. For tripartite systems of local 
dimensions $3$ and $4$, we obtained all SLOCC and LU classes by 
developing new theoretical and numerical methods based on the 
original concept of MEBs.

Some open questions are worth to mention. In the multiqubit case,
hypergraph states are a special case of LME states; it would be 
interesting to generalize the class of LME states to arbitrary 
dimensions and see if a similar relation holds. Nonlocal properties of qudit hypergraph
states were not a part of this work and deserve a separate consideration. Finally, possible applications of these
states as a resource for quantum computing should be
investigated.

\section{Acknowledgments}

FESS is thankfull to M. Hofmann for suggesting the problem 
of extending hypergraph states to qudits.
 For discussions and comments, we acknowledge 
A. Asadian, C. Budroni, M. Gachechiladze, M. Huber, M. Grassl and 
T.M. Kraft. 

The work has been supported by the program Science
without Borders from the Brazilian agency CAPES, by 
CNPq, through the program PVE/PDJ (Pesquisador Visitante
Especial/P\'os-Doutorado  Junior),  Project  No.  401230/
2014-7, Grant No. 150733/2015-1, the  DAAD, 
the EU (Marie Curie CIG  293993/ENFOQI), the ERC 
(Consolidator Grant 683107/TempoQ),
the FQXi Fund (Silicon  Valley  Community Foundation) 
and the DFG.

\section{Appendix}

\subsection{Phase-space picture}

Infinite-dimensional systems are often described through position 
and momentum operators $Q$ and $P$ in a phase-space picture. 
Displacements in this quantum phase-space are performed by unitaries 
\begin{eqnarray}
D(q,p) = e^{i(pQ-qP)}
\end{eqnarray}
where $q$ and $p$ are real numbers. These unitaries satisfy 
$D(q,p)D(q',p')=e^{-i(qp'-pq')}D(q',p')D(q,p)$, characterizing 
a faithful representation of the Heisenberg-Weyl Lie group. 
Performing the transformations
\begin{eqnarray}
Q\rightarrow \kappa Q+\lambda P; \\
P\rightarrow \mu Q +\nu P, 
\end{eqnarray}
subjected to the condition $\kappa\nu-\lambda\mu=1$, will bring 
the Heisenberg-Weyl group into itself. In other words, unitaries 
that perform these transformations will generate the normalizer 
of the Heisenberg-Weyl group. This group constitutes the so-called 
symplectic group in continuous variables and is related to important
concepts in quantum optics such as squeezing. 

In finite-dimensional systems  it is possible to give an analogous 
description in terms of a discrete phase-pace, whenever the Hilbert 
space dimension is a power of a prime number \cite{vourdas}. A general 
displacement in this discrete phase-space is then performed by an operator $D(m,n)=\omega^{mn2^{-1}}X^nZ^m$; the set of these displacement operators 
form an unitary representation of the discrete Heisenberg-Weyl 
group through the multiplication rule
\begin{eqnarray*}
D(m,n)D(m',n')=\omega^{(m'n-mn')2^{-1}}D(m+m',n+n')
\end{eqnarray*}
Symplectic transformations in a discrete-phase space act over the 
Pauli operators in the following fashion
\begin{eqnarray}
SXS^{\dagger} &=& \omega^{\kappa\lambda 2^{-1}}X^{\kappa}Z^{\lambda}, \\
SZS^{\dagger} &=& \omega^{\mu\nu 2^{-1}}X^{\mu}Z^{\nu},
\end{eqnarray}
subjected to the condition $\kappa\nu-\lambda\mu=1 \ (mod \ d)$. 
An arbitrary symplectic operator $S(\kappa,\lambda,\mu)$ can be decomposed as  
\begin{eqnarray}
S(\kappa,\lambda,\mu) = S(1,0,\xi_1)S(1,\xi_2,0)S(\xi_3,0,0)
\end{eqnarray}
where the operators in the right-hand side are given in (\ref{clifford}), (\ref{clifford2}), (\ref{clifford3}) and 
\begin{eqnarray}
\xi_1 &=& \mu\kappa(1+\lambda\mu)^{-1}; \\
\xi_2 &=& \mu\kappa^{-1}(1+\lambda\mu); \\
\xi_3 &=& \kappa(1+\lambda\mu)^{-1}
\end{eqnarray}
The actions of gates (\ref{clifford2}) and (\ref{clifford3}) are respectively given by
\begin{eqnarray}
S(1,\xi,0)XS(1,-\xi,0) &=& \omega^{\xi 2^{-1}}XZ^{\xi}; \\ 
S(1,\xi,0)ZS(1,-\xi,0) &=& Z; \\
S(1,0,\xi)ZS(1,0,-\xi) &=& \omega^{-\xi 2^{-1}}ZX^{\xi}; \\ 
S(1,0,\xi)XS(1,0,-\xi) &=& X.
\end{eqnarray}

\subsection{Local complementation of qudit graphs}

The graph operation known as local complementation of a graph $G=(V,E)$ at the vertex $a\in V$ consists of the following mapping:
\begin{eqnarray}
G\rightarrow G' = (V,E\uplus E_{N_a})
\end{eqnarray}
where $E_{N_a}$ are the edges in the neighbourhood of $a$ and $\uplus$ denotes the set operation of symmetric sum, i.e., $A\uplus B = \{A\cup B\}\setminus\{A\cap B\}$. The implementation of such operation for qudit graph states is known in the literature \cite{quditgraph} and is restricted to prime-dimensional systems. Here we give a simpler derivation of this implementation, which is also valid for some special cases in non-prime dimensional systems. Note that we consider here only graphs with edges of multiplicity one, which are equivalent to graphs with edges of multiplicity coprime with the underlying dimension $d$.

From the section on stabilizers of a hypergraph state, we get as a special case that the operators 
\begin{eqnarray}
K_i = X_i\prod_{e\in E^*}Z_{e\setminus\{i\}}=X_iZ_{N_i}, \ \ i\in V
\end{eqnarray}
generate the stabilizer group of the graph state $|G\rangle$ represented by the graph $G=(V,E)$. The graph state $|G\rangle$ is thus the unique $+1$ eigenstate of the operators $K_i$. 

\begin{theorem}

Given a graph state $|G\rangle$ composed of edges with multiplicity one, let $U_a=S_a(1,0,-1)S_{N_a}(1,-1,0)$. Then, 
$U_a|G\rangle=|G'\rangle$, where $G'$ is the local complementation of $G$ at the vertex $a\in V$.
\end{theorem}
\begin{proof} Let $\{K_i\}_{i\in V}$ denote the set of stabilizer operators of $G$ and let $\mathcal{S}$ be the stabilizer group generated by them. 
It is clear that $U_aK_iU_a^{\dagger}=K_i$ if $i$ is not in $N_a$, 
while for $c\in N_a$ we have $U_aK_cU_a^{\dagger}=K_a^{-1}K'_c$, where $K'_c$ is the stabilizer operator 
for the vertex $c$ of $G'$. We have then that $U_a\mathcal{S}U_a^{\dagger}=\mathcal{S}'$, where $\mathcal{S}'$ is the group generated by the stabilizer operators of $G'$.  \end{proof}

Another way of proving this result is to consider the action of $S_a(1,0,-1)$, which is simply
\begin{eqnarray}
S_a(1,0,-1)|G\rangle= S_{N_a}(1,1,0)|G'\rangle
\end{eqnarray}
and thus applying $S_{N_a}(1,-1,0)$ on the state above will map $G$ into its local complementation $G'$.  

\subsection{Proofs of Proposition~\ref{propclifford}}
\begin{proof}
We define $k = \alpha g$ and $k'= \beta g$ and consider first a single particle
gate $Z$. Looking at the action of $S$, $S^\dagger$ and $Z$ on a basis vector 
$\ket{x}$ one sees that a corresponding $S$ can be found, iff we can find an 
$\xi$ such that 
\begin{equation}
\frac{k x}{\xi} = k' x  \mod d 
\end{equation}
holds for any $x$. Dividing by $g$, this is equivalent to 
$
\alpha x  = \beta x \xi \mod (d/g).
$
The value of $\xi$ is found by considering first 
$\xi' = \alpha / \beta \mod (d/g)$. This is well defined, 
since $\beta$ and $d/g$ are coprime. It remains to construct 
a final $\xi$ that is coprime with $d$. The $\xi'$ fulfills
$\beta \xi' = \alpha + y  (d/g)$ for some $y$. It follows that 
$\xi'$ does not have any prime factors already contained in 
$(d/g)$ since $\alpha$ and $d/g$ are coprime, but $\xi'$ may
still have prime factors present in $g$ (but absent in $d/g$). 
If this is the case, we choose $\xi = \xi' + d/g$. This is allowed, 
since $\xi'$ was defined $\mod (d/g)$. Now, $\xi$ has no prime
factors contained in $g$ (but absent in $d/g$), and still no prime
factors contained in $d/g$. So, it is coprime to $d$, and $S$ is 
unitary. 

Finally, if a multiparticle gate $Z_e$ is considered, the proof is
the same, starting from the representation in Eq.~(\ref{phasegate}).
\end{proof}
\noindent Below is an alternative proof of Proposition~\ref{propclifford}.
\begin{proof} For $gcd(d,k)=1$, i.e., $k$ is coprime with $d$, there exists $k^{-1}$ such that $kk^{-1}=1$; this multiplicative inverse is given by $k^{-1}=k^{\lambda(d)-1}$, where  $\lambda(d)$ is the Carmichael function \cite{carmichael}. 
The function $f:\mathbb{Z}_d\rightarrow \mathbb{Z}_d$ given by $f(q)=qk$ is injective \cite{footnote}, since $qk=q'k$ iff $q=q'$. But it is also surjective since it is a function from $\mathbb{Z}_d$ to itself. 
Hence, $f$ is a bijection, the unitary $S_{k}=\sum_{q=0}^{d-1}|q\rangle\langle qk|$ is well-defined and $S_{k^{-1}}ZS^{\dagger}_{k^{-1}}=Z^k$; notice that this corresponds to the Clifford gate (\ref{clifford}). 
Defining the Clifford operator $S=S_{k'^{-1}}S^{\dagger}_{k^{-1}}$, it follows that $SZ^kS^{\dagger}=S_{k'^{-1}}(S^{\dagger}_{k^{-1}}Z^kS_{k^{-1}})S^{\dagger}_{k'^{-1}}=S_{k'^{-1}}ZS^{\dagger}_{k'^{-1}}=Z^{k'}$.

If $gcd(d,k)=g>1$, then there exists $c$ coprime with $d$ such that $k=gc$. 
In order to prove this, let us take the prime decomposition of $d$, i.e., 
$d=p_1^{n_1}p_2^{n_2}\ldots p_N^{n_N}$. 
Let us consider the decomposition \cite{niven,hardy} $\mathbb{Z}_d\approx\mathbb{Z}_{d_1}\times\mathbb{Z}_{d_2}\times\ldots\mathbb{Z}_{d_N}$, 
where $d_i=p_i^{n_i}$.
Under this decomposition, any $m\in\mathbb{Z}_d$ is expressed as 
$m=(m_1,m_2,\ldots,m_N)$, where $m_i=m(mod \ d_i)$.
Given $c$ coprime with $d$, it is straightforward that $c=(c_1,c_2,\ldots,c_N)$, where each component 
$c_i$ is coprime with $d_i$; indeed, these values $c$ are formed by all possible different combinations of the component values $c_i$. 
The only values $k=g\alpha$ which are not already in the form $k=gc$ can only happen for $g=p_j^{n_j}$, 
for some fixed $j$ and $\alpha=p_j^n$, for fixed $1\leq n\leq n_j$. 
The decomposition of $k=g\alpha$ is simply $k=(k_1,k_2,\ldots,k_j=0,\ldots,k_N)$ and the non-zero values $k_i$, $i\neq j$, are coprime with $d_i$. 
Thus, the values $c$ coprime with $d$ for which $c_i=k_i$ for all $i\neq j$ yields $gc=g\alpha=k$. 

Let $gcd(d,k)=gcd(d,k')=g$; then there exists $c$, $c'$ coprime with $d$ such that $k=gc$ and $k'=gc'$, by the discussion above.
Let $S$ be the Clifford operator such that $SZ^cS^{\dagger}=Z^{c'}$. 
Then $SZ^kS^{\dagger}=S(Z^{c})^gS^{\dagger}=(Z^{c'})^g=Z^{k'}$. \end{proof}


\begin{thebibliography}{99}
\bibitem{footnote}
A more elementary way of proving is noticing that $qk=q'k$ iff $q=q'$ $mod(d/gcd(d,k))$.
\end{thebibliography}
\end{document}